\documentclass[lettersize,journal]{IEEEtran}

\usepackage{cite}
\usepackage{amsmath,amssymb,amsfonts}
\usepackage{amsthm}
\usepackage{mathtools}
\usepackage{textcomp}
\usepackage{graphicx}
\usepackage{xcolor}
\usepackage{microtype}

\usepackage[hidelinks]{hyperref}

\theoremstyle{remark}
\newtheorem{remarkx}{Remark}
\newenvironment{remark}
{\pushQED{\qed}\remarkx}
{\popQED\endremarkx}

\theoremstyle{definition}
\newtheorem{defn}{Definition}
\newtheorem{assump}{Assumption}

\theoremstyle{plain}
\newtheorem{theorem}{Theorem}

\newtheorem{proposition}{\indent \it Proposition}

\newtheorem*{claim*}{Claim}

\newcommand{\defeq}{:=} 
\newcommand{\diag}[1]{{\rm diag}\{ #1\}}

\newcommand{\dist}{{\rm dist}}
\newcommand{\matr}[1]{\begin{bmatrix} #1 \end{bmatrix}}

\newcommand{\mbr}[1][{}]{\mathbb{R}^{#1}}	

\begin{document}
\title{Vector-Field-Guided Constraint-Following Control for Geometric Path Following of Uncertain Mechanical Systems}
\author{Hui Yin, \IEEEmembership{Member, IEEE}, Xiang Li, Yifan Liu, Weijia Yao, \IEEEmembership{Member, IEEE}, Yaonan Wang, \IEEEmembership{Senior Member, IEEE}
	\thanks{This work was supported by the National Natural Science Foundation of China under Grant 52475098.(\emph{Corresponding author: Weijia Yao.})}
	\thanks{Hui Yin is with the State Key Laboratory of Advanced Design and Manufacturing Technology for Vehicle, Hunan University, Changsha, 410082, China(e-mail: huiyin@hnu.edu.cn).}
	\thanks{Yifan Liu is with KUKA Robotics, Foshan, China, 528311, China(e-mail: lyf5400@hnu.edu.cn).}
	\thanks{Xiang Li, Weijia Yao and Yaonan Wang are with the School of Artificial Intelligence and Robotics, Hunan University, Changsha, 410082, China(e-mail: lixiang2000@hnu.edu.cn; wjyao@hnu.edu.cn; yaonan@hnu.edu.cn).}
	\thanks{This work has been submitted to the IEEE for possible publication. Copyright may be transferred without notice, after which this version may no longer be accessible.}
}

\maketitle

\begin{abstract}
 Achieving dynamics-level control is a critical yet challenging problem
for accurate, high-speed geometric path-following of uncertain robotic systems. This paper proposes a vector-field-guided constraint-following control framework for geometric path following of uncertain mechanical systems. The key idea is to encode vector-field guidance as a velocity constraint, enabling a unified force/torque control design for a class of fully actuated and underactuated systems with desired paths that may be self-intersecting. A global, explicit error estimate connects constraint-following accuracy to path-following performance with convergence guarantees. An adaptive robust extension further addresses heterogeneous time-varying uncertainties with unknown bounds and ensures ultimately bounded path-following error. Applications to an underactuated aircraft system and a fully actuated manipulator system demonstrate the effectiveness and advantages of the framework.

\end{abstract}


\section{Introduction}
Geometric path following is a fundamental task in robotics: a system is required to approach a prescribed geometric path and progress along it without a prescribed time schedule. This freedom distinguishes path following from trajectory tracking and can alleviate performance limitations associated with time-parameterized reference motions \cite{aguiar2008performance}. For mechanical systems, realizing such geometric motion requires control forces or torques acting through the system dynamics. Explicit treatment of these dynamics is particularly valuable for accurate and robust motion execution: compensation for dynamical effects can improve tracking accuracy during high-speed maneuvers \cite{faessler2018differential}, while dynamics-level uncertainty compensation can preserve tracking performance under various perturbations \cite{tal2021accurate,8490723,9900371}. These considerations motivate a general path-following design method that connects path geometry with dynamics-level control while accounting for uncertainties.

Existing research provides several approaches to this problem, but their applicability depends on the system dynamics, path geometry, and available model information. For example, the path-following controller in \cite{consolini2010path} addresses the planar vertical take-off and landing (PVTOL) aircraft and a class of smooth Jordan curves, with guarantees on path invariance and roll dynamics. The immersion and invariance-based approach in \cite{yi2020} formulates path following for a class of underactuated mechanical systems as orbital stabilization. These studies establish important dynamics-level results. Nevertheless, a unified design procedure accommodating both fully actuated and underactuated systems, paths that may be open, closed, or self-intersecting, and heterogeneous time-varying uncertainties remains insufficiently developed. Such a procedure would allow different path-following tasks to be addressed through a common control structure, reducing the need for system- and path-specific redesign.

Guiding vector fields (GVFs) provide a concise and useful tool to achieve geometric path-following task, with substantial progress in handling complex path geometries, including open, closed, self-intersecting, and higher-dimensional paths \cite{goncalves2010vector,yao2020path,kapitanyuk2017guiding,yao2021singularity}. Constructed from the desired path, a GVF specifies a velocity field that combines motion toward the path with motion along it. This velocity-level description is naturally suited to kinematics control, where velocity commands can be directly assigned as the inputs, with the assumption that the robot-specific inner-loop dynamics control can realize these commands effectively \cite{yao2021singularity}. Accordingly, existing GVF-based frameworks addressing broad classes of paths and robotic systems have predominantly been developed at kinematics level \cite{goncalves2010vector,yao2020path,kapitanyuk2017guiding,yao2021singularity,yao2023distributed,he2025novel,qiao2026curvature}. However, geometric convergence established at kinematics level does not by itself guarantee path-following performance for mechanical systems with uncertain dynamics. The central challenge of extending GVF-based guidance to dynamics control is to design force/torque inputs that drive the actual velocity to follow a velocity field that depends on the system position, under multi-source of uncertainties. A rigorous solution must further establish how the remaining velocity-following error affects convergence to, and deviation from, the desired path.

Constraint-following control (CFC) provides a dynamics-level framework for enforcing prescribed servo constraints through forces/torques \cite{bajodah2005inve,chen2009cons,Udwadia2015Constrained}, with robust and adaptive extensions accommodating uncertainties \cite{chen2010AdaptiveRobustApproximatea,9741754}. Although CFC has been applied to vehicle path tracking \cite{yang2021path}, extending it to general geometric paths presents a fundamental constraint-design challenge. In particular, a conventional construction imposes a stable differential equation on the functions defining the desired path. The resulting constraint matrix, however, depends on the path Jacobian and may lose rank, compromising constraint feasibility and the associated convergence guarantees. Such difficulties arise even for elementary geometries such as simple-closed circles, and even before motion along the path is considered. Moreover, enforcing a geometric constraint alone does not ensure continued progression along the path. These limitations motivate a systematic servo-constraint design that supports both convergence toward and motion along the desired path and admits dynamics-level enforcement within the CFC framework.

Comparing GVF and CFC, we note an inherent connection between them. That is, the two approaches address complementary parts of the same task: GVF specifies the geometric motion, whereas CFC provides a means of realizing prescribed motion constraints through uncertain mechanical dynamics. The key question is how to connect them through a constraint construction that preserves the path-following objective and admits a common dynamics-level control design. This article gives an answer, proposing a general dynamics control framework for geometric path-following of uncertain mechanical systems. The main contributions are twofold. 

\begin{itemize}
    \item A unified dynamics-level path-following framework called vector-field-guided constraint-following control (VFCFC) is developed by constructing a vector-field-guided constraint (VFC) that connects GVF and CFC. It accommodates fully actuated and underactuated mechanical systems satisfying the stated conditions, together with geometric paths that may be open, closed, or self-intersecting. The relationship between VFC-following error and path-following error is established, providing the theoretical connection between geometric guidance and dynamical control. A nominal controller is then constructed to ensure vanishing path-following error for the nominal system. 
    \item An adaptive robust realization of the framework is developed for systems with heterogeneous time-varying uncertainties whose bounds are unknown. Under the specified uncertainty conditions, the controller ensures uniformly ultimately bounded path-following error without requiring a bound on the rate of uncertainty variation. 
\end{itemize}
Simulations on an underactuated PVTOL aircraft and a fully-actuated space manipulator illustrate the applicability of the proposed framework to different mechanical systems and evaluate its path-following performance through comparisons with existing methods.

The remainder of this article is organized as follows. Section \ref{sec_Preliminaries} reviews necessary preliminaries and formulates the problem. Section \ref{sec_doa} presents the main theoretical results. Section \ref{sec_sim} provides the simulation results. Section \ref{sec_conclusion} concludes the paper.

{\it Notation}: We denote by $\mathbb{R}$ and $t$ the set of reals and the time, respectively. We use $\dot{x}(t)$ and $\ddot{x}(t)$ to respectively represent $\frac{{\rm d}}{{\rm d} t}x(t)$ and $\frac{{\rm d}^2}{{\rm d} t^2}x(t)$ for any $C^2$ function $x$ of $t$. We denote by $\dist(p_0,\mathcal{S}):=\inf\{\|p_0-p_{\rm s}\|:p_{\rm s}\in\mathcal{S}\}$ the distance between a point $p_0 \in \mathbb{R}^n$ and a nonempty set $\mathcal{S}\subseteq \mathbb{R}^n$, and denote by $\dist(\mathcal{S}_1,\mathcal{S}_2):=\inf\{\|p_{\rm s1}-p_{\rm s2}\|:p_{\rm s1}\in\mathcal{S}_1,p_{\rm s2}\in\mathcal{S}_2\}$ the distance between two nonempty sets $\mathcal{S}_1$ and $\mathcal{S}_2$. Consequently, the distance between a system trajectory $x(t):[0,\infty)\to \mathbb{R}^n$ and a desired path $\mathcal{P}\subseteq \mathbb{R}^n$ is $\dist(x(t),\mathcal{P}):=\inf\{\|x(t)-p\|:p\in\mathcal{P}\}$, i.e., the path-following error at $t$.

\section{Preliminaries and problem formulation} \label{sec_Preliminaries}
In this section, we first review some necessary preliminaries in Subsections \ref{sec21}, \ref{sec22}, and \ref{sec23}. Then, the control problem is formulated in subsection \ref{sec_problemform}, which also provides a further explanation for the motivation of our solution strategy.
\subsection{Uncertain mechanical systems}\label{sec21}
\label{sec2}
	The dynamics of a controlled uncertain mechanical system can be generally described as
	\begin{equation}\label{mechanicalsystem}
 \setlength{\abovedisplayskip}{0pt}
\setlength{\belowdisplayskip}{0pt}
	\begin{split}
	 &M(q(t),\sigma(t),t)\ddot{q}(t)+C(q(t),\dot{q}(t),\sigma(t),t)\dot{q}(t)\\
  &+ g(q(t),\sigma(t),t)=B(q(t),\sigma(t),t)\tau(t).
	 \end{split}
	\end{equation}
	Here $q\in\mathbb{R}^n$ is the configuration vector; $t\geq t_0\geq 0$ with $t_0$ the initial time; $\tau \in \mathbb{R}^{m}$ ($m\leq n$) is the control force/torque; $\sigma\in\Sigma \subseteq\mathbb{R}^p$ denotes uncertainty with $\Sigma$ an unknown but compact set. The inertia matrix $M$ is assumed to be positive definite and thus invertible, which holds for the vast majority of practical applications \cite{chen2010AdaptiveRobustApproximatea}; $C\dot{q}$, $g$ are the Coriolis/centrifugal and gravitational forces respectively, $B$ is the input matrix. The functions $M:\mathbb{R}^n\times\Sigma\times\mathbb{R}\to\mathbb{R}^{n\times n}$, $C:\mathbb{R}^{n}\times\mathbb{R}^{n}\times\Sigma\times\mathbb{R}\to\mathbb{R}^{n\times n}$, $g:\mathbb{R}^n\times\Sigma\times\mathbb{R}\to\mathbb{R}^{n}$, $B:\mathbb{R}^n\times\Sigma\times\mathbb{R}\to\mathbb{R}^{n\times m}$ are all sufficiently smooth such that the existence and uniqueness of solutions to \eqref{mechanicalsystem} is guaranteed for any given initial condition $(q(t_0),\dot{q}(t_0))$ \cite[Theorem 3.1]{khalil2002nonlinear}. The system \eqref{mechanicalsystem} represents underactuated  systems when $m<n$. 
	
\begin{remark}\label{remarkuncertainty}
 We do not assume the uncertainty $\sigma$ to be fully matched or to only include sufficiently small mismatched uncertainty \cite{zhang2022}. It may enter all the system matrices $M$, $C$, $g$, and $B$, allowing for uncertain parameters and state- and time-dependent
force perturbations. Its bound is not required to be
known, and no bound on $\dot{\sigma}(t)$ is imposed. Thus, the formulation permits rapidly time-varying
uncertainties and does not restrict them a priori
to the control-input channels.
\end{remark}

\subsection{Constraint-following control}\label{sec22}
Constraint-following control (CFC) designs control forces
or torques to enforce prescribed servo constraints on
mechanical systems
\cite{chen2009cons,Udwadia2015Constrained, chen2010AdaptiveRobustApproximatea,9741754}.
Consider a servo constraint of the form
\begin{equation}
    A(q,t)\dot q=c(q,t),
    \label{matrixform1}
\end{equation}
where $A(q,t)\in\mathbb{R}^{m\times n}$ and
$c(q,t)\in\mathbb{R}^{m}$ are continuously differentiable.
The constraint-following error is defined as
\begin{equation}
    \beta(q,\dot q,t):=A(q,t)\dot q-c(q,t).
    \label{eq:constraint-error}
\end{equation}
For system \eqref{mechanicalsystem}, \emph{exact} constraint-following with respect to \eqref{matrixform1} corresponds to
$\beta(q(t),\dot q(t),t)\equiv0$, whereas \emph{asymptotic} and \emph{approximate}
constraint-following refer to $\beta(q(t),\dot q(t),t)\to0$ as $t\to\infty$ and
uniform ultimate boundedness of $\beta(q(t),\dot q(t),t)$, respectively.
A control input achieving the corresponding property
for \eqref{mechanicalsystem} is called a CFC.

Taking time derivative of \eqref{eq:constraint-error} yields
\begin{equation}
    \dot\beta(q(t),\dot q(t),t)=A(q(t),t)\ddot q(t)-b(q(t),\dot q(t),t),
    \label{eq:constraint-error-dynamics}
\end{equation}
where $b=\dot c-\dot A\dot q$ with $\dot A$ and $\dot c$ denoting total time
derivatives. Hence, the acceleration constraint
associated with \eqref{matrixform1} is
\begin{equation}
    A(q,t)\ddot q=b(q,\dot q,t).
    \label{matrixform2}
\end{equation}
Enforcing \eqref{matrixform2} preserves
$\beta$, so exact satisfaction of
\eqref{matrixform1} additionally requires
$\beta(q(t_0),\dot q(t_0),t_0)=0$. Initial constraint deviations must
instead be addressed through error feedback.

\begin{assump}\label{feasbility}
 For each $(q,\dot{q},t)\in \mathbb{R}^n\times\mathbb{R}^n\times \mathbb{R}$, there is at least one solution $\ddot{q}\in \mathbb{R}^n$ to the constraint \eqref{matrixform2}.  
\end{assump}
\begin{remark}
  Assumption \ref{feasbility} simply implies that in principle the designed constraint has to be feasible, so that the constraint-following task is possible. A necessary and sufficient condition to verify Assumption \ref{feasbility} is, the equation $A(q,t)A^+(q,t)b(q,\dot{q},t) = b(q,\dot{q},t)$ holds for all $(q,\dot{q},t)\in \mathbb{R}^n\times\mathbb{R}^n\times \mathbb{R}$, with the superscript ``$+$" denoting the Moore-Penrose generalized inverse. This is an indispensable assumption to rule out impossible constraint-following tasks.
\end{remark}

It has been shown that stabilization and trajectory-tracking can be encoded by the proposed constraint form \cite{chen2009cons}. However, it is unclear yet how to use this form to accommodate the path-following task in a general theoretical framework.

\subsection{Guiding vector field}\label{sec23}
Since this study employs the singularity-free
GVF developed in \cite{yao2021singularity}, this subsection
reviews only this construction. Suppose a physical desired path $\mathcal{P} \subseteq \mathbb{R}^m$ that may be open, closed, or self-intersecting is described by the intersection of $(m-1)$ hypersurfaces
\begin{equation} \label{mdpath}
 \setlength{\abovedisplayskip}{0pt}
\setlength{\belowdisplayskip}{0pt}
		\mathcal{P} = \{ \zeta\defeq[q_{1},\dots, q_{m}]^{\top} \in \mathbb{R}^m : \psi_i(\zeta) = 0, i=1,\cdots, m-1 \},
\end{equation}
which can be parameterized by
\begin{equation*} 
 \setlength{\abovedisplayskip}{0pt}
\setlength{\belowdisplayskip}{0pt}
	q_{1} = f_1(w), \quad \dots ,\quad q_{m}=f_m(w).
\end{equation*}
Here, $\zeta $ is the physical coordinate, $f_i\in C^2$ with $f'_i(w):=\frac{{\rm d} f_i(w)}{{\rm d} w}$,  $i=1,\cdots,m,$ and $w \in \mbr[]$ is the parameter. Let $\phi_i(q_{1},\dots,q_{m},w)\defeq q_{i}-f_i(w)$ ($i=1,\cdots,m$) and $\xi \defeq [\zeta^{\top},w]^{\top}$ with $w$ being viewed as an additional virtual coordinate. Then we can obtain a higher dimensional desired path in the $(m+1)$-dimensional $\xi$-space as
\begin{equation*}
 \setlength{\abovedisplayskip}{0pt}
\setlength{\belowdisplayskip}{0pt}
	\mathcal{P}^{\rm hgh} = \{ \xi \in \mathbb{R}^{m+1} : \phi_i(\xi) = 0, i=1,\dots m \}.
\end{equation*}
\begin{remark}\label{markphgh}
    Intuitively, the $(m+1)$-D path $\mathcal{P}^{\rm hgh}$  is obtained by “cutting” and “stretching” the $m$-D path $\mathcal{P}$ along the virtual $w$-axis. Inversely speaking, the path $\mathcal{P}$ is the orthogonal projection of the path $\mathcal{P}^{\rm hgh}\subseteq\mathbb{R}^{m+1}$ in the space $\mathbb{R}^m$ with the remaining $w$-axis as zero. 
We can also obtain that, for a virtual system trajectory $\xi(t)$ living in the same $(m+1)$-D space as $\mathcal{P}^{\rm hgh}$, its orthogonal projection in the $m$-D space where $\mathcal{P}$ lives, is the physical system trajectory $\zeta(t)$. It has been shown in \cite{yao2021singularity} that the path-following errors dist$(\xi(t),\mathcal{P}^{\rm hgh})$ and dist$(\zeta(t),\mathcal{P})$ have the same
convergence behaviour. Therefore, we can investigate dist$(\zeta(t),\mathcal{P})$ by dist$(\xi(t),\mathcal{P}^{\rm hgh})$. 
\end{remark}

 The gradients of $\phi_i$ ($i=1,\cdots,m$) are $\nabla \phi_i=[0,\cdots,1,\cdots,-f'_i(w)]^\top$, where ``$1$" is the $i$-th element of $\nabla \phi_i$.
Then by \cite[Eq. (4)]{yao2021singularity}, the GVF for $\mathcal{P}^{\rm hgh}$ is 
\begin{equation}	\label{eq_vfpf_general}
 \setlength{\abovedisplayskip}{0pt}
\setlength{\belowdisplayskip}{0pt}
	\chi(\xi) = \matr{ (-1)^m f'_1(w) - k_1 \phi_1(\xi) \\  \vdots \\ (-1)^m f'_m(w) - k_m \phi_m(\xi) \\ (-1)^m + \sum_{i=1}^{m} k_i \phi_i(\xi) f'_i(w)} \in \mbr[m+1],
\end{equation}
where $k_i\in \mathbb{R}_+, i=1,\cdots,m$ are constants. The first and second terms on the RHS of \eqref{eq_vfpf_general} are the propagation and converging terms, respectively. They are always linearly independent unless they all equal zero, which is impossible due to the term $(-1)^m$ in the last element. Therefore, the GVF $\chi$ is singularity-free, i.e., $\chi(\xi)\neq 0$ for any $\xi\in \mathbb{R}^{m+1}$. This appealing property of $\chi$ is attributed to the introduction of the additional dimension $w$, and thus the assumption of no singular points that may limit the range of admissible desired paths becomes unnecessary. Therefore, this article will adopt only the singularity-free GVF, despite the fact that all the results can also be applied to the general GVF in \cite{yao2020path}.

By the GVF $\chi$, the trajectories of the following system:
\begin{equation} \label{eq_odegvf}
 \setlength{\abovedisplayskip}{0pt}
\setlength{\belowdisplayskip}{0pt}
	\dot{\xi}(t) = \chi(\xi(t))
\end{equation}
globally asymptotically converge to and move along the desired path  without requiring the system origin to be stable \cite[Theorem 2]{yao2021singularity}, \cite[Corollary 1]{yao2023topo}, \cite{he2025nonholo}. Therefore, if one can design a control to drive the trajectory of a mechanical system to satisfy \eqref{eq_odegvf}, then the path-following task can be completed. 

\subsection{Problem formulation} \label{sec_problemform}
\subsubsection{Problem statement}
Consider uncertain mechanical systems with dynamic model \eqref{mechanicalsystem}, the control problem is to design a state feedback control $\tau$ that drives the system trajectories to follow the desired path $\mathcal{P}\subseteq\mathbb{R}^m$ defined like \eqref{mdpath}:
\begin{equation} \label{mdpath4sys}
 \setlength{\abovedisplayskip}{0pt}
\setlength{\belowdisplayskip}{0pt}
		\mathcal{P} = \{ q^{\rm s} \in \mathbb{R}^m : \psi_i(q^{\rm s}) = 0, i=1,\cdots, m-1 \},
\end{equation}
where $q^{\rm s}\in \mathbb{R}^m$ comprises some or all entries of $q \in \mathbb{R}^n$, not necessarily being $[q_{1},\dots, q_{m}]^{\top}$.


\subsubsection{Motivated solution strategy}\label{motivation} 
In the CFC approach, the constraint is specified as \eqref{matrixform1} to describe control tasks. Using \eqref{matrixform1}, conventional methods reformulate desired path \eqref{mdpath4sys} into the following equivalent form: $\dot{\psi}(q^{\rm s})+\Lambda\psi(q^{\rm s})=0 $ with $\Lambda\succ 0$ a constant matrix. Such methods pose two issues. First, the resulting matrix $A$ includes the derivative of $\psi(q^{\rm s})$ with respect to $q^{\rm s}$. This indicates that $A$ will be state-dependent once $\psi(q^{\rm s})$ is nonlinear in $q^{\rm s}$, which may lead to the violation of Assumption \ref{feasbility} (as  illustrated by an example in the simulation section). Second, the resulting CFC can only guarantee system trajectories to follow the above constraint to approach the desired path, but cannot guarantee system trajectories to move along the desired path.

In the GVF-based control approach, given the GVF \eqref{eq_vfpf_general}, the remaining task is to design a control such that the dynamics of $\xi$ associated with a mechanical system satisfies  \eqref{eq_odegvf}. However, current GVF-based approaches use kinematic models of robotic systems, assuming an effective system-specific inner-loop dynamics control, and cannot be directly applied to the dynamic model \eqref{mechanicalsystem} of mechanical systems. 

 Comparing ODEs \eqref{matrixform1} and \eqref{eq_odegvf}, we can find that \eqref{eq_odegvf} can be cast into the form of \eqref{matrixform1}, with a constant (state-independent) matrix $A$. This implies a bridge between CFC and GVF. Given the complementary strengths and limitations of CFC and GVF in addressing path-following tasks, bridging them suggests a promising solution to the stated control problem. The realization of this idea represents the first contribution claimed in the introduction, which is detailed in subsequent sections. 

\section{Main results} \label{sec_doa}
\subsection{Vector-field-guided constraint}
Regarding the GVF $\chi$ used in this study, i.e., \eqref{eq_vfpf_general}, $\xi$ includes a virtual position coordinate $w$, and \eqref{eq_odegvf} can be rewritten into two parts: one corresponding to $\dot{\zeta}$, and the other corresponding to $\dot{w}$. The part of $\dot{\zeta}$ is
\begin{equation} \label{zetapart}
 \setlength{\abovedisplayskip}{0pt}
\setlength{\belowdisplayskip}{0pt}
	\dot{\zeta}(t) = \chi^{\rm s}(\xi(t)),
\end{equation}
where
\begin{equation*}	
 \setlength{\abovedisplayskip}{0pt}
\setlength{\belowdisplayskip}{0pt}
	\chi^{\rm s}(\xi) = \matr{ (-1)^m f'_1(w) - k_1 \phi_1(\xi) \\  \vdots \\ (-1)^m f'_m(w) - k_m \phi_m(\xi) } \in \mbr[m].
\end{equation*}
The part of $\dot{w}$ is
\begin{equation}\label{wpart}
 \setlength{\abovedisplayskip}{0pt}
\setlength{\belowdisplayskip}{0pt}
    \dot{w}=(-1)^m + \sum_{i=1}^{m} k_i \phi_i(\xi) f'_i(w) \in \mbr[].
\end{equation}

As $q^{\rm s}\in \mathbb{R}^m$ consists of some (or all) entries of $q \in \mathbb{R}^n$ in \eqref{mechanicalsystem}, there exists a linear relation between $q$ and $q^{\rm s}$. Namely, there exists a constant matrix $A \in \mathbb{R}^{m \times n}$ of full row rank that can be constructed in a general way, such that $A q = q^{\rm s}$. A method to construct $A$ is presented as follows. Suppose that the vector $q^{\rm s}=[q_{s_1},q_{s_2},\cdots,q_{s_m}]^\top$, with $s_1<s_2<\cdots<s_m$ and $s_i\in \{1,2,,\cdots,n\}, i=1,\cdots,m$. Note that $q=[q_1,q_2,\cdots,q_n]^\top$. Then the element at the $i$-th row and $j$-th column in the $m\times n$ matrix $A$ satisfying $Aq=q^{\rm s}$ is \begin{equation*}
a_{ij}=\begin{cases}
1, & \mbox{if   }   s_i=j, \\
0, & \mbox{otherwise}.
\end{cases}
\end{equation*}

Let $\zeta\defeq q^{\rm s}$, then $\xi=(\zeta^{\top},w)^{\top}$, and the part of $\dot{\zeta}$ \eqref{zetapart} can be re-written as follows
\begin{equation} \label{VFC}
 \setlength{\abovedisplayskip}{0pt}
\setlength{\belowdisplayskip}{0pt}
	\dot{\zeta}(t)=A \dot{q}(t) = \chi^{\rm s}(A q(t),w(t)).
\end{equation}
 Comparing \eqref{matrixform1} and \eqref{VFC}, one immediately obtains that \eqref{VFC} can be viewed as the constraint in the CFC framework. Namely, the CFC can be used to enable the trajectory $q(t)$ of uncertain mechanical system \eqref{mechanicalsystem}  to satisfy constraint \eqref{VFC} and eventually follow the desired path encoded in \eqref{VFC}. As \eqref{VFC} is constructed based on the vector field, we call it the \emph{vector-field-guided constraint (VFC)}, and the resulting CFC is called the \emph{vector-field-guided constraint-following control (VFCFC)}. The second order form of \eqref{VFC} can be cast into \eqref{matrixform2} with $A$ being the same as that in \eqref{VFC} and  
\begin{equation*}
 \setlength{\abovedisplayskip}{0pt}
\setlength{\belowdisplayskip}{0pt}
\begin{split}
 b(q,\dot{q},w)&=\dot{\chi^{\rm s}}(q,w) \stackrel{\eqref{wpart}}{=} \frac{\partial \chi^{\rm s}(q,w) }{\partial q}\dot{q}\\
 &+\frac{\partial \chi^{\rm s}(q,w) }{\partial w}\left[(-1)^m + \sum_{i=1}^{m} k_i \phi_i(\xi) f'_i(w)\right]. 
 \end{split}
\end{equation*} 
\begin{remark}
 Note that the part of $\dot{w}$  \eqref{wpart} is not included in the VFC, since \eqref{wpart} is a designed virtual system which as an equation always holds for each real-time measurement of $q$. Nevertheless, \eqref{wpart} will serve as a virtual ``adjoint system” governing the evolution of $w$ to ``shape" the VFC \eqref{VFC} dynamically for the control design. Concretely, for any given initial value $w(t_0)$ and each real-time measured $q$, the control algorithm simply needs to substitute $q$ into \eqref{wpart} and solve it to obtain $w$. The resulting $w$ can then be used in the VFC \eqref{VFC} as the desired system motion being the control goal.
In a word, to fulfill the path-following task, we only need to
design a control to drive the motion of system \eqref{mechanicalsystem} to follow \eqref{VFC}
with a virtual ``adjoint system” \eqref{wpart} governing $w$, which is similar to the adaptive law governing adaptive parameters in adaptive controls \cite{ioannou2006adap}.   
\end{remark}

\begin{remark}
It is worth emphasizing that the matrix $A$ of the VFC is constant and of full row rank, which brings essential benefits to the CFC design for path-following. The most salient one is, with such a matrix $A$, the fundamental Assumption \ref{feasbility} to rule out impossible constraint-following tasks is always satisfied. While in conventional CFC design, Assumption \ref{feasbility} has to be cautiously verified. Furthermore, the matrix $A$ is generally involved in some other assumptions within the CFC framework, e.g., \cite{chen2010AdaptiveRobustApproximatea} assuming $A$ to be of full row rank. A constant $A$ of full row rank, instead of a state-dependent $A$, will greatly ease the restrictions arising from these assumptions.
\end{remark}

 \subsection{Error analysis}
In this subsection, error analysis will be implemented to assure the effectiveness of VFC for encoding
the path-following task. From the above discussion, we can find that there are three error signals of the actual system trajectory. 

The first error signal is the \emph{path-following error}: 
\[
    \dist(\xi(t),\mathcal{P}^{\rm hgh}) = \inf\{\|\xi(t)-p\|:p\in\mathcal{P}^{\rm hgh}\}.
\]
Rendering dist($\xi(t),\mathcal{P}^{\rm hgh}$) to converge to $0$ or to be (ultimately) bounded as $t\to \infty$ is the ultimate control goal. However, dist($\xi(t),\mathcal{P}^{\rm hgh}$) is generally computationally complicated. The second error signal is the \emph{path-function error}:
\[
    \|\phi(\xi)\|=\|[\phi_1(\xi),\phi_2(\xi),\cdots,\phi_{m}(\xi)]^{\top}\|,    
\]
which is more straightforward for computation.  It is introduced to ``represent'' the path-following error dist($\xi,\mathcal{P}^{\rm hgh}$), which is supported by Proposition \ref{prop:geometric-distance-bound}.
\begin{proposition}[Geometric distance bound]
\label{prop:geometric-distance-bound}
For the higher-dimensional desired path
$\mathcal{P}^{\rm hgh}$ and the path-function error
$\phi(\xi)$ defined above, it holds that
\begin{equation}
\label{eq:geometric-distance-bound}
\operatorname{dist}\bigl(\xi,\mathcal{P}^{\rm hgh}\bigr)
\leq \|\phi(\xi)\|,
\qquad \forall \xi\in\mathbb{R}^{m+1}.
\end{equation}
\end{proposition}

\begin{proof}
Let $\xi=(\zeta^\top,w)^\top$ and
$f(v)=[f_1(v),\ldots,f_m(v)]^\top$. 
Thus, every point on $\mathcal{P}^{\mathrm{hgh}}$ can be
represented by $p(v):=[f(v)^\top,v]^\top$ for some $v\in\mathbb{R}$.
By the definition of $\operatorname{dist}\bigl(\cdot,\cdot\bigr)$,
\[
\begin{aligned}
\operatorname{dist}\bigl(\xi,\mathcal{P}^{\mathrm{hgh}}\bigr)
=\inf_{v\in\mathbb{R}}
\sqrt{\|\zeta-f(v)\|^2+|w-v|^2}.
\end{aligned}
\]
Particularly, choosing $v=w$ gives the path point
$p(w)$, which has the same virtual
coordinate as $\xi$. Since the infimum is no greater
than the value at this particular choice, we obtain
\[
\begin{aligned}
\operatorname{dist}\bigl(\xi,\mathcal{P}^{\mathrm{hgh}}\bigr)
&\leq
\sqrt{\|\zeta-f(w)\|^2+|w-w|^2}\\
&=\|\zeta-f(w)\|=\|\phi(\xi)\|.
\end{aligned}
\] \end{proof}
\begin{remark}
Proposition~\ref{prop:geometric-distance-bound} shows that
$\|\phi(\xi(t))\|\to 0$ implies
$\operatorname{dist}(\xi(t),\mathcal{P}^{\mathrm{hgh}})\to 0$.
Likewise, any bound on
$\|\phi(\xi(t))\|$ also bounds $\operatorname{dist}(\xi(t),\mathcal{P}^{\mathrm{hgh}})$. These implications follow from the higher-dimensional path construction. This result strengthens the relationship between the path-function error and the geometric distance posited in the first part of Assumption~2 of \cite{yao2021singularity},  eliminating the need to impose this assumption. More specifically, unlike the assumption therein, which only delivers convergence to zero from $\|\phi(\xi(t))\|$ to $\operatorname{dist}(\xi(t),\mathcal{P}^{\mathrm{hgh}})$, the explicit bound established here delivers both convergence to zero and boundedness. 
\end{remark}


The third error signal is the \emph{VFC-following error}:
\[
    \beta(q,\dot{q},w)=A\dot{q}-\chi^{\rm s}(Aq,w),  
\]
which can be directly used to indicate the satisfaction of VFC \eqref{VFC}. We will show in the sequel that the convergence behavior of $\|\beta\|$ aligns with that of $\|\phi\|$. Therefore, the path-following control design task can be fulfilled by analyzing $\|\beta\|$, as long as certain consistency of convergence behaviors between the path-following error dist($\xi,\mathcal{P}^{\rm hgh}$) and the path-function error $\|\phi(\xi)\|$ is guaranteed. 

The VFC-following error $\beta$ indicates the distance between $\dot{\zeta}$
and $\chi^{\rm s}$, which affects the satisfaction of \eqref{VFC}. Furthermore, as the part of $\dot{w}$ \eqref{wpart} always holds, it does not bring any errors. Consequently, we can obtain the disturbed system of \eqref{eq_odegvf} as
\begin{equation*}
 \setlength{\abovedisplayskip}{0pt}
\setlength{\belowdisplayskip}{0pt}
    \dot{\xi}(t)=\chi(\xi(t))+d(t)
\end{equation*}
with $d(t) = [\beta^{\top}(q(t),\dot{q}(t),w(t)), 0]^{\top} \in \mathbb{R}^{m+1}$ being the disturbance.
Then the governing equation of $\phi$  is
\begin{equation}\label{errordyn}
 \setlength{\abovedisplayskip}{0pt}
\setlength{\belowdisplayskip}{0pt}
    \dot{\phi}(\xi(t))=N^{\top}(\xi(t))(\chi(\xi(t))+d(t))
\end{equation}
with $N(\xi)=[\nabla \phi_1(\xi),\nabla \phi_2(\xi),\cdots,\nabla\phi_m(\xi)]\in \mathbb{R}^{(m+1)\times m}$.
With the governing equation \eqref{errordyn} of $\phi$, the following proposition demonstrates the feasibility of using the convergence behavior of $\|\beta\|$ to imply the corresponding convergence behavior of $\text{dist}(\xi,\mathcal{P}^{\rm hgh})$. 
\begin{proposition}[Constraint-to-path error relationship]
\label{prop:constraint-to-path-error}
Consider the system \eqref{errordyn}. For every piecewise-continuous, uniformly bounded input
$d(t) = [\beta^{\top}(q(t),\dot{q}(t),w(t)), 0]^{\top}$, this system is forward complete, and the
path-following error $\operatorname{dist}(\xi(t),\mathcal{P}^{\mathrm{hgh}})$ is uniformly ultimately bounded.
Moreover,
\begin{equation*}
\limsup_{t\to\infty}
\operatorname{dist}(\xi(t),\mathcal{P}^{\mathrm{hgh}})
\leq \frac{1}{k_{\min}}
\limsup_{t\to\infty}\|\beta(q(t),\dot{q}(t),w(t))\|,
\end{equation*}
where $k_{\min}=\min_i k_i$.
Thus, $\beta(q(t),\dot{q}(t),w(t))\to0$ implies $\operatorname{dist}(\xi(t),\mathcal{P}^{\mathrm{hgh}})\to 0$.
\end{proposition}

\begin{proof}
For brevity, write $\phi(t)=\phi(\xi(t))$, $\beta(t)=\beta(q(t),\dot{q}(t),w(t))$ and
$f'(w)=[f_1'(w),\ldots,f_m'(w)]^\top$. Let $V=\frac12\phi^\top K\phi$ and
$a=f'^\top K\phi$ with $K=\mathrm{diag} \left(k_1,\cdots,k_m\right)$.
By Young's inequality, $K^2\succeq k_{\min}K$, and \eqref{errordyn},
\begin{equation}
\label{eq:path-error-dissipation}
\begin{aligned}
\dot V
&=-(K\phi)^\top N^\top N(K\phi)
  +(K\phi)^\top\beta\\
&=-\|K\phi\|^2-|a|^2+(K\phi)^\top\beta\\
&\leq-\tfrac12\|K\phi\|^2-|a|^2
  +\tfrac12\|\beta\|^2\\
&\leq-k_{\min}V+\tfrac12\|\beta\|^2.
\end{aligned}
\end{equation}
Consequently, for any $t\geq t_0$ in the interval
of existence,
\begin{equation}
\label{eq:path-error-comparison}
V(t)\leq e^{-k_{\min}(t-t_0)}V(t_0)
+\frac12\int_{t_0}^{t}
e^{-k_{\min}(t-\tau)}\|\beta(\tau)\|^2\,d\tau.
\end{equation}

Let $\|\beta(t)\|\leq b_{\rm sup}<\infty$, then 
the estimate \eqref{eq:path-error-comparison} gives a bound on $\phi$.
Retaining the term $-|a|^2$ in
\eqref{eq:path-error-dissipation} also gives
\begin{equation}
\label{eq:at}
\int_{t_0}^{t}|a(\tau)|^2\,d\tau
\leq V(t_0)+\tfrac12 b_{\rm sup}^2(t-t_0).
\end{equation}
Since $\dot w=(-1)^m+a$, combing the Cauchy--Schwarz
inequality and \eqref{eq:at} yields
\[
|w(t)-w(t_0)|
\leq t-t_0+
\sqrt{(t-t_0)\bigl[V(t_0)+\tfrac12 b_{\rm sup}^2(t-t_0)\bigr]}.
\]
Thus, neither $w$ nor $\zeta=f(w)+\phi$ can escape
in finite time. Local regularity then ensures forward
completeness.

Using $V\geq k_{\min}\|\phi\|^2/2$ in
\eqref{eq:path-error-comparison} gives
\[
\|\phi(t)\|^2
\leq \frac{2V(t_0)}{k_{\min}}e^{-k_{\min}(t-t_0)}
+\frac{b_{\rm sup}^2}{k_{\min}^2}
\bigl(1-e^{-k_{\min}(t-t_0)}\bigr).
\]
Hence every bound strictly greater than $b_{\rm sup}/k_{\min}$
is a uniform ultimate bound for $\|\phi\|$; and for any prescribed bounded set of initial errors, there exists a common finite entry time for all initial conditions in that set, independent of $t_0$.

Finally, assume that there exists a finite value
$b_\infty=\limsup_{t\to\infty}\|\beta(t)\|$.
For any $\epsilon>0$, choose $t_1$ sufficiently large
such that $\|\beta(t)\|\leq b_\infty+\epsilon$ for $t\geq t_1$.
Applying~\eqref{eq:path-error-comparison} from $t_1$
and letting $t\to\infty$ gives
\[
\limsup_{t\to\infty}\|\phi(t)\|
\leq (b_\infty+\epsilon)/k_{\min}.
\]
Letting $\epsilon\to0$ and invoking
Proposition~\ref{prop:geometric-distance-bound}
proves the stated geometric error bound and the
vanishing-error conclusion.
\end{proof}

\begin{remark}
Proposition \ref{prop:constraint-to-path-error} strengthens the local ISS result
in \cite[Proposition~5]{yao2021singularity} by providing
a global, explicit constraint-to-path error estimate,
without restrictions on the initial path-function
error or the VFC-following error. It also removes the boundedness assumption on $f'(w)$
from both the error analysis and the forward-completeness
proof for the extended kinematic system.
The key is the global inequality
$N^\top N=I_m+f'(w)f'(w)^\top\succeq I_m$.
Together with Proposition~\ref{prop:geometric-distance-bound},
this establishes both vanishing and ultimate-boundedness
guarantees for the path-following error, with the
explicit asymptotic gain $1/k_{\min}$ from $\beta$ to the
path distance.    
\end{remark}


Based on the above results, the control design for the path-following problem of the mechanical system \eqref{mechanicalsystem} proceeds in two steps. First, neglecting the uncertainty, a class of nominal path-following controls is designed. Second, considering the uncertainty, a class of adaptive robust path-following controls is designed based on the nominal control. 

\subsection{Nominal path-following control}
The nominal system of $\eqref{mechanicalsystem}$ can be described as
\begin{equation}\label{nominalsystem}
 \setlength{\abovedisplayskip}{0pt}
\setlength{\belowdisplayskip}{0pt}
	\begin{split}
	 \bar{M}(q(t),t)\ddot{q}(t)+\bar{C}(q(t),\dot{q}(t),t)\dot{q}(t)\\+\bar{g}(q(t),t)=\bar{B}(q(t),t)\tau(t).
	 \end{split}
	\end{equation}
In this study, the nominal and uncertain portions of a system matrix $(\cdot)$ are denoted by $\bar{(\cdot)}$ and $\Delta (\cdot)$  respectively, hence $(\cdot)=\bar{(\cdot)}+\Delta (\cdot)$. For example, $\Bar{M}$ represents the nominal portion of $M$, which is also chosen to be positive definite. The nominal path-following control that steers the nominal system \eqref{nominalsystem} to follow $\mathcal{P}$ is designed under the following assumption.
\begin{assump}\label{assump_tau2}
   For any given constant matrix $P\in \mathbb{R}^{m \times m}$, $P\succ 0$, there exists a positive constant $\underline{\lambda}>0$ such that $\lambda_m(PA\bar{M}^{-1}(q,t)\bar{B}(q,t) (A\bar{M}^{-1}(q,t)\bar{B}(q,t))^{\top}P) \ge \underline{\lambda}$ for all $(q,t)\in \mathbb{R}^n \times\mathbb{R}_{\geq 0}$, where $\lambda_m(\cdot)$ denotes the minimal eigenvalue of a matrix.
\end{assump}

\begin{remark}
   Note $PA\bar{M}^{-1}(q,t)\bar{B}(q,t) (A\bar{M}^{-1}(q,t)\bar{B}(q,t))^{\top}P$ is positive semi-definite for all $(q,t)\in \mathbb{R}^n \times\mathbb{R}$. Assumption \ref{assump_tau2} simply requires that the minimum eigenvalue remain bounded away from zero. Another implication of  Assumption \ref{assump_tau2} is that the matrix $A\bar{M}^{-1}(q,t)\bar{B}(q,t)$ is invertible for all $(q,t)\in \mathbb{R}^n \times\mathbb{R}$, a property that does not automatically follow even when $A$ and $B$ are full rank (as the product may lose rank). This assumption excludes tasks where the matrix $A$ results in a non-invertible $A\bar{M}^{-1}(q,t)\bar{B}(q,t)$ for some $(q,t)\in \mathbb{R}^n\times\mathbb{R}$, as its inverse is required in the control design. A similar assumption can be found in existing CFC studies, e.g., \cite{yin2019ControllingUnderactuatedTwoWheeleda, yin2020tackling}. However, the distinction lies in the fact that Assumption \ref{assump_tau2} in this study uses a constant $A$ that is of full row rank, whereas the corresponding assumption in existing CFC studies involves an \( A \) that depends on \( (q,t) \). This difference stems from the proposed VFC, which greatly relaxes the assumption restriction.
\end{remark}

The nominal path-following control is proposed as
\begin{equation}\label{nominal}
 \setlength{\abovedisplayskip}{0pt}
\setlength{\belowdisplayskip}{0pt}
    \tau(t)=p_1(q(t),\dot{q}(t),w(t),t)+p_2(q(t),\dot{q}(t),w(t),t),
\end{equation}
with
\begin{equation}\label{nom_p1}
 \setlength{\abovedisplayskip}{0pt}
\setlength{\belowdisplayskip}{0pt}
\begin{split}
p_1(q,\dot{q},w,t)=(A\bar{M}^{-1}(q,t)\bar{B}(q,t))^{-1}[b(q,\dot{q},w)\\+A\bar{M}^{-1}(q,t)(\bar{C}(q,\dot{q},t)\dot{q}+\bar{g}(q,t))],\\
p_2(q,\dot{q},w,t)=-\kappa (A\bar{M}^{-1}(q,t)\bar{B}(q,t))^{\top}P\beta(q,\dot{q},w),
\end{split}
\end{equation}
and
$\beta=A\dot{q}-\chi^{\rm s}$, where $\kappa\in\mathbb{R}_+$ is a gain parameter. 

\begin{theorem}[Vanishing error for nominal systems]
	Under Assumption \ref{assump_tau2}, the control \eqref{nominal} drives the trajectory of nominal mechanical system \eqref{nominalsystem} to converge to the desired path \eqref{mdpath4sys}, i.e., the path-following error $\dist(\zeta(t), \mathcal{P})\to 0$ as $t\to\infty$.
\end{theorem}
\begin{proof}
	By the results in \cite[Theorem 2]{yin2019ControllingUnderactuatedTwoWheeleda}, it follows that $\beta(q(t),\dot{q}(t),w(t))=A\dot{q}(t)-\chi^{\rm s}(Aq(t),w(t))$ exponentially converges to $0$ as $t\to\infty$. Then by Proposition \ref{prop:constraint-to-path-error}, since the VFC-following error converges to 0 as $t\to\infty$, the path-following error dist$(\xi(t),\mathcal{P}^{\rm hgh})$ converges to $0$ as $t\to \infty$, and thus by Remark \ref{markphgh}, we have $\dist(\zeta(t), \mathcal{P})\to 0$ as $t\to\infty$. 
\end{proof}

\subsection{Adaptive robust path-following control}
To robustly achieve path-following for the uncertain system \eqref{mechanicalsystem} with the desired path \eqref{mdpath4sys}, this section designs an adaptive robust action $p_3$ added to the nominal control \eqref{nominal} to suppress uncertainty, resulting in a class of adaptive robust path-following controls. Due to the uncertainties, vanishing path-following error is not always possible, while the best achievable performance is typically the boundedness of path-following error, which is the objective of designing $p_3$.

To handle the mismatched uncertainty, we perform the following uncertainty decompositions. Let $H:=\Bar{M}M^{-1}-I$, which is employed to represent the uncertainty in $M$. Let
\begin{equation}\label{decomp}
 \setlength{\abovedisplayskip}{0pt}
\setlength{\belowdisplayskip}{0pt}
 H(q,\sigma,t)=\Bar{B}(q,t)\check{H}(q,\sigma,t) +\tilde{H}(q,\sigma,t). 
\end{equation}
 This is to divide $H$ into two portions based on $\Bar{B}$, i.e., the matched uncertainty $\Bar{B}\check{H}$ and the mismatched uncertainty $\tilde{H}$, which lies within and outside the range space of $\Bar{B}$, respectively. Throughout the paper, we use $\check{(\cdot)}$ and $\tilde{(\cdot)}$ to respectively denote matched and mismatched portions. There are infinite choices for the particular expressions of  $\check{H}$ and $\tilde{H}$. For example, one may choose any 
$\check{H}$ and let $\tilde{H}:=H-\Bar{B}\check{H}$. However, we choose $\check{H}$ and then obtain $\tilde{H}$ by \eqref{decomp} as follows
\begin{equation} \label{eq32}
 \setlength{\abovedisplayskip}{0pt}
\setlength{\belowdisplayskip}{0pt}
\begin{split}
 &\check{H}=(A\bar{M}^{-1}\bar{B})^{-1}A\bar{M}^{-1}H,\\   
 &\tilde{H}=H-\Bar{B}(A\bar{M}^{-1}\bar{B})^{-1}A\bar{M}^{-1}H,
 \end{split}
\end{equation}
where recall that $A\bar{M}^{-1}\bar{B}$ is invertible by Assumption \ref{assump_tau2}. Similar to \eqref{decomp}, other uncertain system matrices, $\Delta C$, $\Delta g$, and $\Delta B$,  can also be decomposed based on the nominal matrix $\Bar{B}$. The detailed expressions for these decompositions are omitted here for simplicity.  From \eqref{eq32}, we obtain the following identity regarding the mismatched uncertainty $\tilde{H}$
\begin{equation}\label{uncertainty_identity}
 \setlength{\abovedisplayskip}{0pt}
\setlength{\belowdisplayskip}{0pt}
    A\bar{M}^{-1}\tilde{H}=A\bar{M}^{-1}[H-\bar{B}(A\bar{M}^{-1}\bar{B})^{-1}A\bar{M}^{-1}H]\equiv 0.
\end{equation}
Similarly, it can be shown that the mismatched uncertainties of other uncertain system matrices satisfy analogous identities, i.e.,  $A\bar{M}^{-1}\tilde{(\cdot)}\equiv 0$, where $(\cdot)$ represents $C$, $g$ or $B$.

\begin{assump}\label{assump5}
There exists a (possibly unknown) constant $\rho_{W} >-1$ such that for all $(q,t)\in \mathbb{R}^{n}\times \mathbb{R}$
\begin{equation*}
 \setlength{\abovedisplayskip}{0pt}
\setlength{\belowdisplayskip}{0pt}
\frac{1}{2} \min\limits_{\sigma\in\Sigma} \lambda_{m}(W(q,\sigma,t)+W^{\top}(q,\sigma,t))\ge \rho_{W}
\end{equation*}
with
\begin{equation*}
 \setlength{\abovedisplayskip}{0pt}
\setlength{\belowdisplayskip}{0pt}
W(q,\sigma,t)=\check{B}(q,\sigma,t)+\check{H}(q,\sigma,t)B(q,\sigma,t).
\end{equation*}
\end{assump}

\begin{remark} 
	Note that $W$ is synthesized by $M$ and $B$. The inequality of $W$ implies that the minimum eigenvalue of symmetric matrix $(W+W^\top)$ over the domain of uncertainty $\sigma$ should not be smaller than a constant $\rho_W>-1$. For extreme cases, if $M$ has no uncertainty, then the condition reduces to $W=\check{B}$; If furthermore, $B$ has no uncertainty, then $W=0$, and the assumption can be removed. Therefore, this assumption simply characterizes how much uncertainties in $M$ and $B$ can be tolerated in the control design, similar to \cite[Assumption 4]{chen2010AdaptiveRobustApproximatea}, \cite[Assumption 4]{yin2020tackling}.
\end{remark}

\begin{assump}\label{assump6}
	There exists a  (possibly unknown) constant vector $\alpha\in (0,\infty)^k$ with $k\in\mathbb{N}_+$ and a known function $\Pi(\cdot):(0,\infty)^k\times \mathbb{R}^n\times \mathbb{R}^n\times \mathbb{R}\times \mathbb{R}\to \mathbb{R}_+$ such that for all $(q,\dot{q},t)\in \mathbb{R}^n\times\mathbb{R}^n\times\mathbb{R}_+$,
	\begin{equation}\label{assumption6}
  \setlength{\abovedisplayskip}{0pt}
\setlength{\belowdisplayskip}{0pt}
		\begin{split}
&(1+\rho_W)^{-1}\underset{\sigma\in\Sigma}{\max}\|\check{H}(q,\sigma,t)(-C(q,\dot{q},\sigma,t)\dot{q}-g(q,\dot{q},\sigma,t)\\
&+B(q,\sigma,t)p_1(q,\dot{q},w,t))+(\check{B}(q,\sigma,t)p_1(q,\dot{q},w,t)\\
&-\check{C}(q,\dot{q},\sigma,t)\dot{q}-\check{g}(q,\dot{q},\sigma,t))\|\leq\Pi(\alpha,q,\dot{q},w,t).
		\end{split}
	\end{equation}
	Furthermore, the function $\Pi$ can be linearly factored with respect to $\alpha$, 
	\begin{equation*}
  \setlength{\abovedisplayskip}{0pt}
\setlength{\belowdisplayskip}{0pt}
		\Pi(\alpha,q,\dot{q},w,t)=\alpha^{\top}\breve{\Pi}(q,\dot{q},w,t),
	\end{equation*}
	where $\breve{\Pi}(\cdot):\mathbb{R}^n\times \mathbb{R}^n\times \mathbb{R}\times \mathbb{R}\to \mathbb{R}_+^k$.
\end{assump}

\begin{remark} 
	Since $\rho_W>-1$, we have $\rho_W+1>0$. Note that the left-hand side of \eqref{assumption6} aggregates all system uncertainties, thus the function $\Pi$ serves as an envelope bound that upper-bounds the aggregated uncertainty term over the uncertainty domain. While the structure of $\Pi$ is known, the constant vector $\alpha$ is possibly unknown since it is relevant to the unknown bounding set $\Sigma$. Therefore, Assumption~\ref{assump6} simply aims to provide a functional envelope by $\Pi$ to characterize the worst-case effect of system uncertainties.
\end{remark}

The adaptive robust path-following control is proposed as:
\begin{equation}\label{arpc}
 \setlength{\abovedisplayskip}{0pt}
\setlength{\belowdisplayskip}{0pt}
\begin{split}
\tau(t)=p_{1}(q(t),\dot{q}(t),w(t),t)+p_{2}(q(t),\dot{q}(t),w(t),t)\\+p_{3}(\hat{\alpha}(t),q(t),\dot{q}(t),w(t),t).
\end{split}
\end{equation}
Here, $p_{1,2}$ are given in \eqref{nominal}, and $p_3$ is designed as
\begin{equation}\label{CFC_p3}
 \setlength{\abovedisplayskip}{0pt}
\setlength{\belowdisplayskip}{0pt}
\begin{split}
p_{3}(\hat{\alpha},q,\dot{q},w,t)=-\eta(\hat{\alpha},q,\dot{q},w,t)\upsilon(\hat{\alpha},q,\dot{q},w,t)\Pi(\hat{\alpha},q,\dot{q},w,t),
\end{split}
\end{equation}
where
\begin{equation*}
 \setlength{\abovedisplayskip}{0pt}
\setlength{\belowdisplayskip}{0pt}
\upsilon=\bar{B}^{\top}\bar{M}^{-1} A^{\top} P \beta \Pi,\quad  \setlength{\abovedisplayskip}{0pt}
\setlength{\belowdisplayskip}{0pt}
\eta=\begin{cases}
\frac{1}{\| \upsilon \|}, & \mbox{if   } \| \upsilon \| > \mu, \\
\frac{1}{\mu}, & \mbox{otherwise},
\end{cases}
\end{equation*}
and $\mu$ is a positive scalar constant. The vector $-\eta\upsilon$ defines the ``direction" to compensate uncertainties. The parameter vector $\hat{\alpha}\in \mathbb{R}^{k}$ is governed by the following adaptive law:
\begin{equation}\label{CFC_zi2}
 \setlength{\abovedisplayskip}{0pt}
\setlength{\belowdisplayskip}{0pt}
\dot{\hat{\alpha}}=\begin{cases}
\ell_1 \breve{\Pi}\|\breve{\beta}\| -\ell_2 \hat{\alpha}, & \mbox{if   }\|\breve{\Pi}\| \|\breve{\beta}\| > \varepsilon,\\
\ell_1 \breve{\Pi} \frac{\| \breve{\Pi}\| \|\breve\beta\|^{2}}{\varepsilon}-\ell_2 \hat{\alpha}, & \mbox{otherwise},
\end{cases}
\end{equation}
where $\varepsilon$, $\ell_{1,2}$ are positive scalar constants, $\hat{\alpha}_{i}(t_{0})>0$ and $\breve{\beta}=\bar{B}^{\top}\bar{M}^{-1}A^{\top}P \beta$. Hence $\upsilon=\breve{\beta} \Pi$.

\begin{remark} 
The piecewise continuous adaptive law \eqref{CFC_zi2} is of leakage type, and the second term on the RHS is the leak. We can obtain $\hat{\alpha}_i(t)>0$ for all $t\geq t_0$ with $\hat{\alpha}_{i}(t_{0})>0$. This is because the first term on the RHS is always non-negative, and by the property of differential inequality, the solution of \eqref{CFC_zi2} is always no less than that of $\dot{\hat{\alpha}}=-\ell_2\hat{\alpha}$ with $\hat{\alpha}_{i}(t_{0})>0$, which is always positive. The adaptive parameter $\hat{\alpha}$ tends to emulate the unknown $\alpha$ in the bound function $\Pi$, which will be proved in Theorem \ref{thm2}. As $p_3$ is designed by $\hat{\alpha}$, it is called adaptive robust action. Consequently, the resulting control $\tau$ consisting of $p_{1,2,3}$ is called adaptive robust path-following control. 
\end{remark}

\begin{theorem}[Boundedness for uncertain systems]\label{thm2}
Under Assumptions \ref{assump_tau2}-\ref{assump6}, the adaptive robust path-following control~\eqref{arpc} ensures that, for the uncertain system \eqref{mechanicalsystem} with desired path \eqref{mdpath4sys}, both the path-following error $\dist(\zeta(t), \mathcal{P})$ and the adaptive error $\|\hat{\alpha}(t)-\alpha\|$ are uniformly ultimately bounded.
\end{theorem}
\begin{proof}
See the Appendix.    
\end{proof}

\begin{remark}
 From the proof, the boundedness region is primarily determined by the gain parameter $\kappa$, the matrix $P$, and the adaptive-law parameters $\ell_{1,2}$. These four factors collectively influence the extent of boundedness. Notably, for fixed $P$ and $\ell_{1,2}$, the expression characterizing boundedness indicates that increasing $\kappa$ tends to shrink---and never enlarges---the boundedness region. Hence, larger $\kappa$ can enhance path-following accuracy. Moreover, although the main results are derived using the singularity-free GVF, the idea can be applied to the general GVF in \cite{yao2020path} without using the virtual state $w$.
\end{remark}

\section{Simulation examples}  \label{sec_sim}
\subsection{An underactuated PVTOL aircraft system}\label{underactuated}
This section presents simulations on an underactuated PVTOL aircraft \cite{consolini2010path} for demonstrations. The corresponding dynamic model is
\begin{equation}\label{PVTOL}
 \setlength{\abovedisplayskip}{0pt}
\setlength{\belowdisplayskip}{0pt}
    \begin{split}
    &m(t)\ddot{x}(t)=-(\sin\theta(t))\tau_1(t)+(\cos\theta(t))\tau_2(t)+d_x(t),\\
    &m(t)\ddot{y}(t)=(\cos\theta(t))\tau_1(t)+(\sin\theta(t))\tau_2(t)-m(t)g_0+d_y(t),\\
    &J(t)\ddot{\theta}(t)=\tau_2(t)+d_\theta(t),
    \end{split}
\end{equation}
where states $(x,y)$ and $\theta$ represent respectively the position and roll angle, control inputs $\tau_1$ and $\tau_2$ are respectively the thrust force and rolling torque, $m$ and $J$ are respectively the mass and the moment of inertia containing time-varying uncertainties, $d_{x,y,\theta}$ are unknown external disturbances, $g_0$ is the gravitational acceleration. The uncertain parameters $m$ and $J$ can be decomposed as $m(t)=\bar{m}+\Delta m(t)$, $J(t)=\bar{J}+\Delta J (t)$. The nominal PVTOL aircraft system can be obtained as
\begin{equation}\label{nomPVTOL}
 \setlength{\abovedisplayskip}{0pt}
\setlength{\belowdisplayskip}{0pt}
    \begin{split}
    &\bar{m}\ddot{x}=-(\sin\theta)\tau_1+(\cos\theta)\tau_2,\\
    &\bar{m}\ddot{y}=(\cos\theta)\tau_1+(\sin\theta)\tau_2-\bar{m}g_0, \quad \bar{J}\ddot{\theta}=\tau_2.
    \end{split}
\end{equation}
Let $q:=[x,y,\theta]^\top$, $\tau=[\tau_1,\tau_2]^\top$. It is straightforward to cast \eqref{PVTOL} and \eqref{nomPVTOL} into the form of \eqref{mechanicalsystem} and \eqref{nominalsystem}, respectively. The particular expressions of system matrices in this case are:
\begin{equation*}
 \setlength{\abovedisplayskip}{0pt}
\setlength{\belowdisplayskip}{0pt}
\begin{split}
    &M=\diag{m,m,J},\ \bar{M}=\diag{\bar{m},\bar{m},\bar{J}},\\
    &C=\bar{C}=0_{3\times 3},\ g=[-d_x,mg_0-d_y,-d_{\theta}]^\top, \\ &\bar{g}=\begin{bmatrix}
            0\\\bar{m}g_0\\0
        \end{bmatrix},\ B=\bar{B}=\begin{bmatrix}
            -\sin\theta & \cos\theta\\
            \cos\theta & \sin\theta\\
            0 & 1
        \end{bmatrix}.
\end{split}
\end{equation*}

Three different desired paths are considered, i.e., the non-closed sinusoidal path $\mathcal{P}_1=\{(x,y)\in \mathbb{R}^2:y=\sin x\}$, the closed Cassini’s oval $(x^2+y^2+ r_a^2)^2-4r_a^2x^2-r_b^4 = 0$ where $r_a=3$, $r_b=1.05r_a$, and the self-intersecting lemniscate of Bernoulli $\mathcal{P}_3=\{(x,y)\in \mathbb{R}^2:(x^2+y^2)^2=9(x^2-y^2)\}$. To apply VFCFC,  $\mathcal{P}_1$, $\mathcal{P}_2$, and $\mathcal{P}_3$ are respectively parameterized as: $\mathcal{P}_1^{\rm hgh}=\{(x,y,w)\in \mathbb{R}^3:x=w,y=\sin w\}$, $\mathcal{P}_2^{\rm hgh}=\{(x,y,w)\in \mathbb{R}^3:x= \cos(\omega) \sqrt{ r_a^2 \cos(2\omega) + \sqrt{ r_b^4 - r_a^4 \sin^2(2\omega) } } ,y=\sin(\omega) \sqrt{ r_a^2 \cos(2\omega) + \sqrt{ r_b^4 - r_a^4 \sin^2(2\omega) } }\}$, $\mathcal{P}_3^{\rm hgh}=\{(x,y,w)\in \mathbb{R}^3: x=\frac{3\cos w}{1+\sin^2 w},\ y=\frac{3\cos w\sin w}{1+\sin^2 w}\}$. As presented in \eqref{VFC}, all three VFCs for $\mathcal{P}_{1,2,3}$ can be readily constructed with $A=\begin{bmatrix}
        1 & 0 & 0\\
       0& 1 & 0
    \end{bmatrix}$ that is of full row rank, and the particular expressions of $\chi^{\rm s}$ are omitted here for simplicity.  Therefore, we have 
\begin{equation*}
 \setlength{\abovedisplayskip}{0pt}
\setlength{\belowdisplayskip}{0pt}
    A\bar{M}^{-1}\bar{B}=\frac{1}{\bar{m}}\begin{bmatrix}
        -\sin\theta & \cos\theta\\
        \cos\theta & \sin\theta
    \end{bmatrix}
\end{equation*}
 for all three paths. Assumptions \ref{feasbility}--\ref{assump_tau2} can be verified for all three paths straightforwardly. For Assumption \ref{assump5}, we have
\begin{equation*}
 \setlength{\abovedisplayskip}{0pt}
\setlength{\belowdisplayskip}{0pt}
\begin{split}
    W&=\check{H}B+\check{B}=(A\bar{M}^{-1}\bar{B})^{-1}A\bar{M}^{-1}(\bar{M}M^{-1}-I)B\\&+(A\bar{M}^{-1}\bar{B})^{-1}A\bar{M}^{-1}\Delta B=\left(\frac{ \bar{m}}{m}-1\right)I.
\end{split}
\end{equation*}
Therefore, Assumption \ref{assump5} can be verified, since \begin{equation*}
 \setlength{\abovedisplayskip}{0pt}
\setlength{\belowdisplayskip}{0pt}
    \frac{1}{2}\min_{\sigma\in\Sigma}\lambda_{m}(W+W^{\top})=\min_{\sigma\in\Sigma}\left(\frac{ \bar{m}}{m}-1\right)>-1.
\end{equation*} 
For Assumption \ref{assump6}, the inequality \eqref{assumption6} is deduced as follows:
\begin{equation*}
 \setlength{\abovedisplayskip}{0pt}
\setlength{\belowdisplayskip}{0pt}
    \begin{split}
        &(1+\rho_W)^{-1}\max_{\sigma\in\Sigma}\|\check{H}(-C\dot{q}-g+Bp_1)+(\check{B}p_1-\check{C}\dot{q}-\check{g})\|\\
        &=(1+\rho_W)^{-1}\max_{\sigma\in\Sigma}\|-\check{H}(C\dot{q}+g)-(\check{C}\dot{q}-\check{g})+(\check{H}B\\
        &\quad+\check{B})p_1)\|\\
        &=(1+\rho_W)^{-1}\max_{\sigma\in\Sigma}\|-(A\bar{M}^{-1}\bar{B})^{-1}A\bar{M}^{-1}H(C\dot{q}+g)\\
        &\quad -(A\bar{M}^{-1}\bar{B})^{-1}A\bar{M}^{-1}(C\dot{q}+g)+Wp_1\|\\
        &=(1+\rho_W)^{-1}\max_{\sigma\in\Sigma}\|-(A\bar{M}^{-1}\bar{B})^{-1}A\bar{M}^{-1}(H+I)(C\dot{q}\\
        &\quad+g)+Wp_1\|\\
        &\le(1+\rho_W)^{-1}\max_{\sigma\in\Sigma}\|(H+I)(C\dot{q}+g)\|\|(A\bar{M}^{-1}\bar{B})^{-1}A\bar{M}^{-1}\|\\
        &\quad+(1+\rho_W)^{-1}\max_{\sigma\in\Sigma}\|W\|\|p_1\|.
    \end{split}
\end{equation*}
Therefore, we choose $\Pi$ as 
\begin{equation*}
 \setlength{\abovedisplayskip}{0pt}
\setlength{\belowdisplayskip}{0pt}
    \Pi(\alpha,q,\dot{q},\omega,t)=\alpha^{\top}\breve{\Pi}(q,\dot{q},\omega,t):=[\alpha_1\;\alpha_2]\begin{bmatrix}
            \breve{\Pi}_1(q,\dot{q},t)\\
            \breve{\Pi}_2(q,\dot{q},\omega,t)
        \end{bmatrix},
\end{equation*}
where $\alpha_1=(1+\rho_W)^{-1}\max_{\sigma\in\Sigma}\|(H+I)(C\dot{q}+g)\|$, $\alpha_2=(1+\rho_W)^{-1}\max_{\sigma\in\Sigma}\|W\|$, $\breve{\Pi}_1=\|(A\bar{M}^{-1}\bar{B})^{-1}A\bar{M}^{-1}\|$, $\breve{\Pi}_2=\|p_1\|$. Both $\alpha_1$ and $\alpha_2$ are unknown.

For simulations, system parameters and uncertainties are chosen as $\bar{m}=1\text{kg}$, $\bar{J}=0.5\text{kg}\cdot \text{m}^2$, $\Delta m(t)=0.3\bar{m}\sin(5t)$, $\Delta J(t)=0.3\bar{J}\sin(7.5t)$, $d_x(t)=4\sin(2t)$, $d_y(t)=4\cos(4t)$, $d_\theta(t)=4\sin(6t)$, $x(0)=2.2 \;\mathrm{m}$, $y(0)=0.2 \;\mathrm{m}$, $\theta(0)=1.5 \;\mathrm{rad}$, $\dot{x}(0)=1 \;\mathrm{m/s}$, $\dot{y}(0)=0 $, $\dot{\theta}(0)=0$. The simulations are implemented in two steps, with the first step testing the nominal VFCFC (NVFCFC) \eqref{nominal} for nominal system \eqref{nomPVTOL}, and the second step testing the adaptive robust VFCFC (ARVFCFC) \eqref{arpc} for uncertain system \eqref{PVTOL}.

In the first step, the parameters of NVFCFC are chosen as $k_1=k_2=1$, $P$ the identity matrix, $\kappa=5$, $\omega(0)=0.1$. The conventional CFC (CCFC), and the immersion and invariance-based orbital stabilization (IIOS) in \cite{yi2020} are also tested for comparisons. When applying CCFC, the indispensable Assumption \ref{feasbility} can be satisfied by the resulting constraint to encode $\mathcal{P}_1$, but cannot be satisfied by the resulting constraint to encode $\mathcal{P}_2$ and $\mathcal{P}_3$. Specifically, by invoking the constraint design method of CCFC presented in ``Motivated solution strategy" of Section \ref{sec_problemform}, one obtains constraints to encode $\mathcal{P}_2$ and $\mathcal{P}_3$ whose second-order forms as \eqref{matrixform2} have no solution for $(\ddot{x},\ddot{y})$ at $(x,y,\dot{x},\dot{y})=(0,0,0,0)$, violating Assumption \ref{feasbility}. On the other hand, IIOS in \cite{yi2020} is only applicable to non-intersecting closed paths, therefore it is only tested on $\mathcal{P}_2$.

\begin{figure}[htbp]
            \centering
            \includegraphics[scale=0.4]{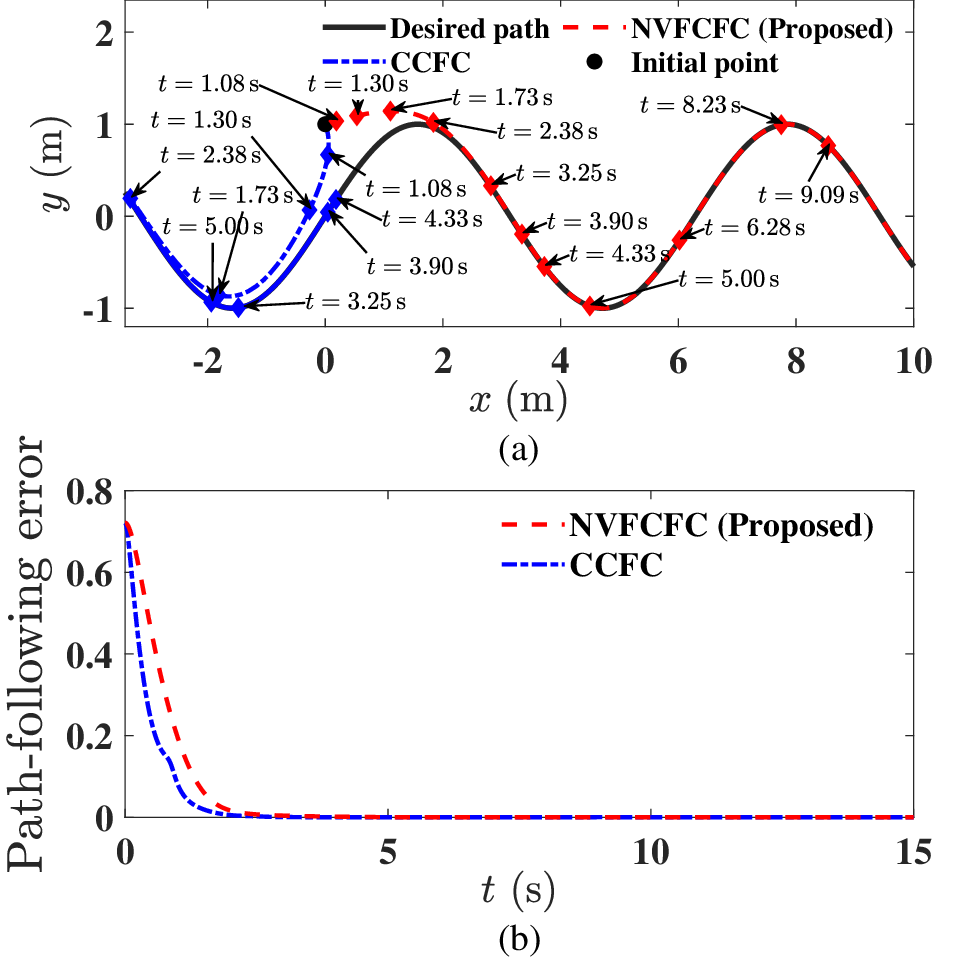}
            \caption{(a) Trajectories and (b) path-following errors of the nominal PVTOL aircraft system following path $\mathcal{P}_1$.}
            \label{cominal_compare_traditional}
\end{figure}

Figs.\ref{cominal_compare_traditional}--\ref{nominal_robust} show the simulation results of the first step. NVFCFC completes all three path-following tasks with vanishing path-following errors. Furthermore, NVFCFC exhibits a certain degree of robustness to system parameter variations (Fig. \ref{nominal_robust}). Despite achieving faster error convergence than NVFCFC, CCFC fails to achieve the path-following task of $\mathcal{P}_1$, since the system trajectory moves back and forth in the neighborhood of some point on the desired path after convergence, failing to traverse the desired path as intended (Fig.\ref{cominal_compare_traditional}). Therefore, CCFC fails any of the three path-following tasks. IIOS only completes the path-following task of $\mathcal{P}_2$, but produces a larger path-following error after convergence than NVFCFC, despite consuming significantly higher control effort  (Figs. \ref{nominal_compare_path} and \ref{nominal_compare_tau}).

\begin{figure}[htbp]
    \centering
            \centering
            \includegraphics[scale=0.3]{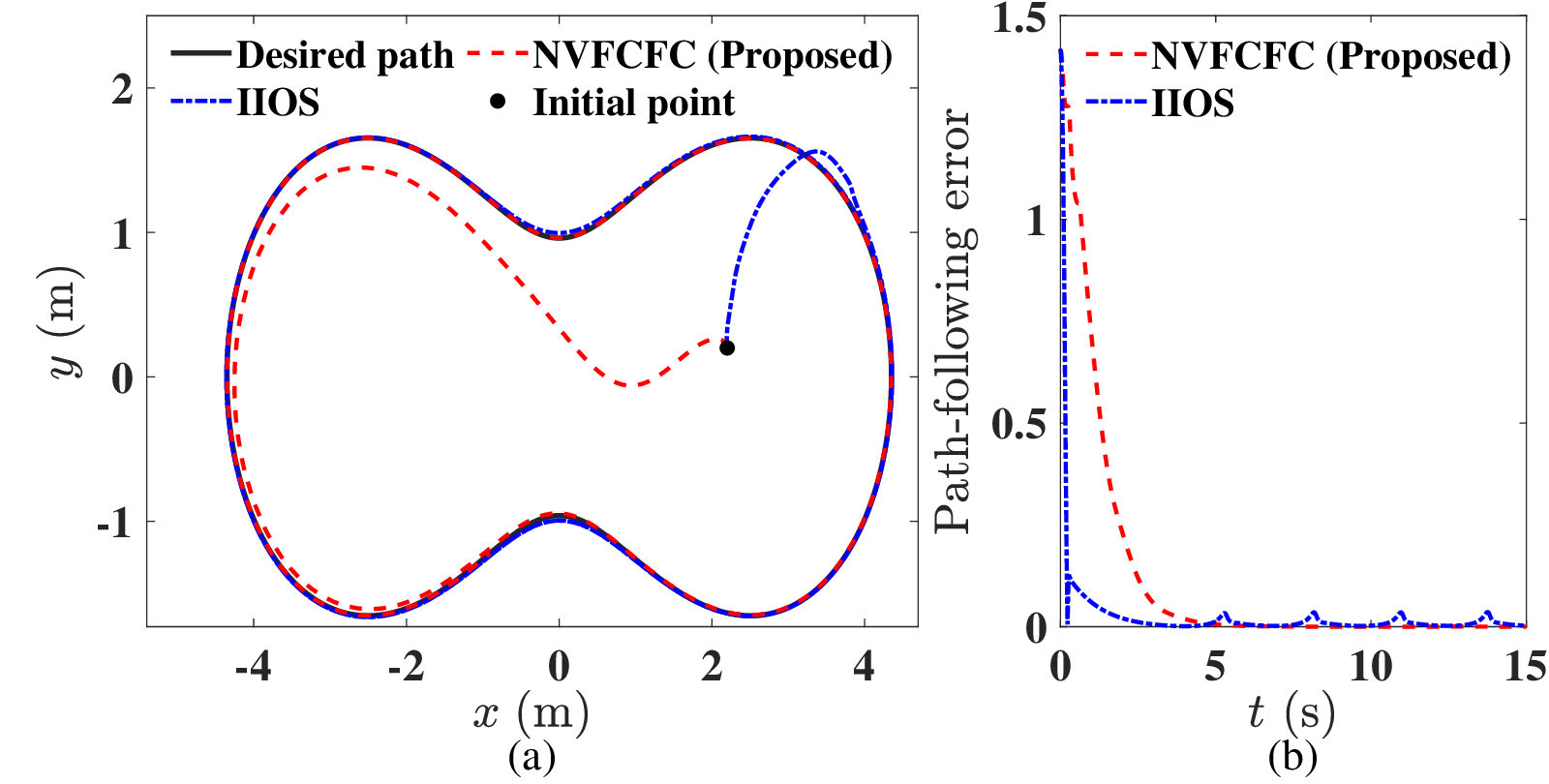}
       \caption{(a) Trajectories and (b) path-following errors of the nominal PVTOL aircraft system following path $\mathcal{P}_2$.}
 \label{nominal_compare_path}
\end{figure}

\begin{figure}[htbp]
            \centering
            \includegraphics[scale=0.3]{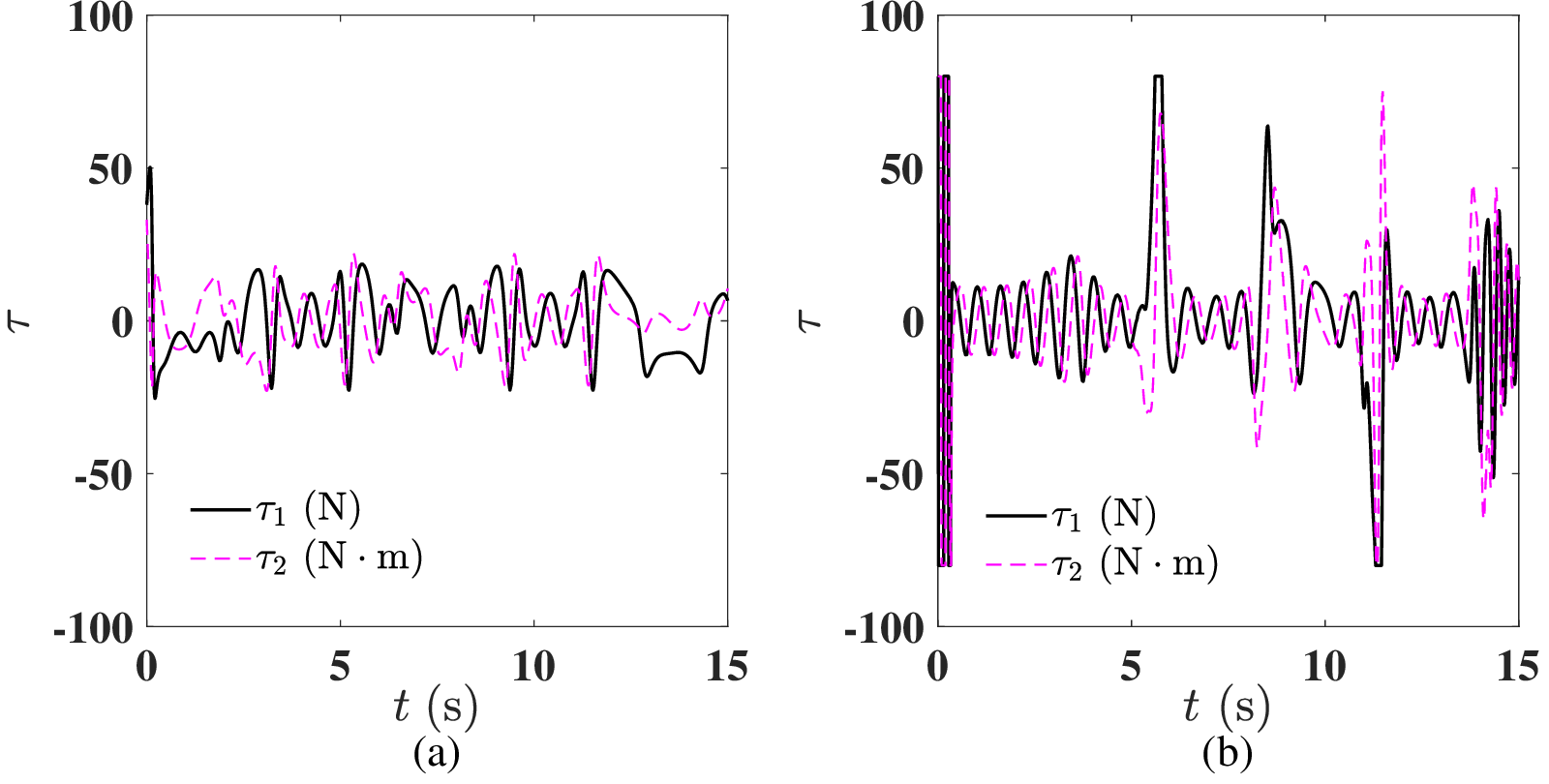}
            \caption{Control inputs of the nominal PVTOL aircraft system following path $\mathcal{P}_2$ with (a) NVFCFC and (b) IIOS.}
            \label{nominal_compare_tau}
\end{figure}

\begin{figure}[htbp]
            \centering
            \includegraphics[scale=0.3]{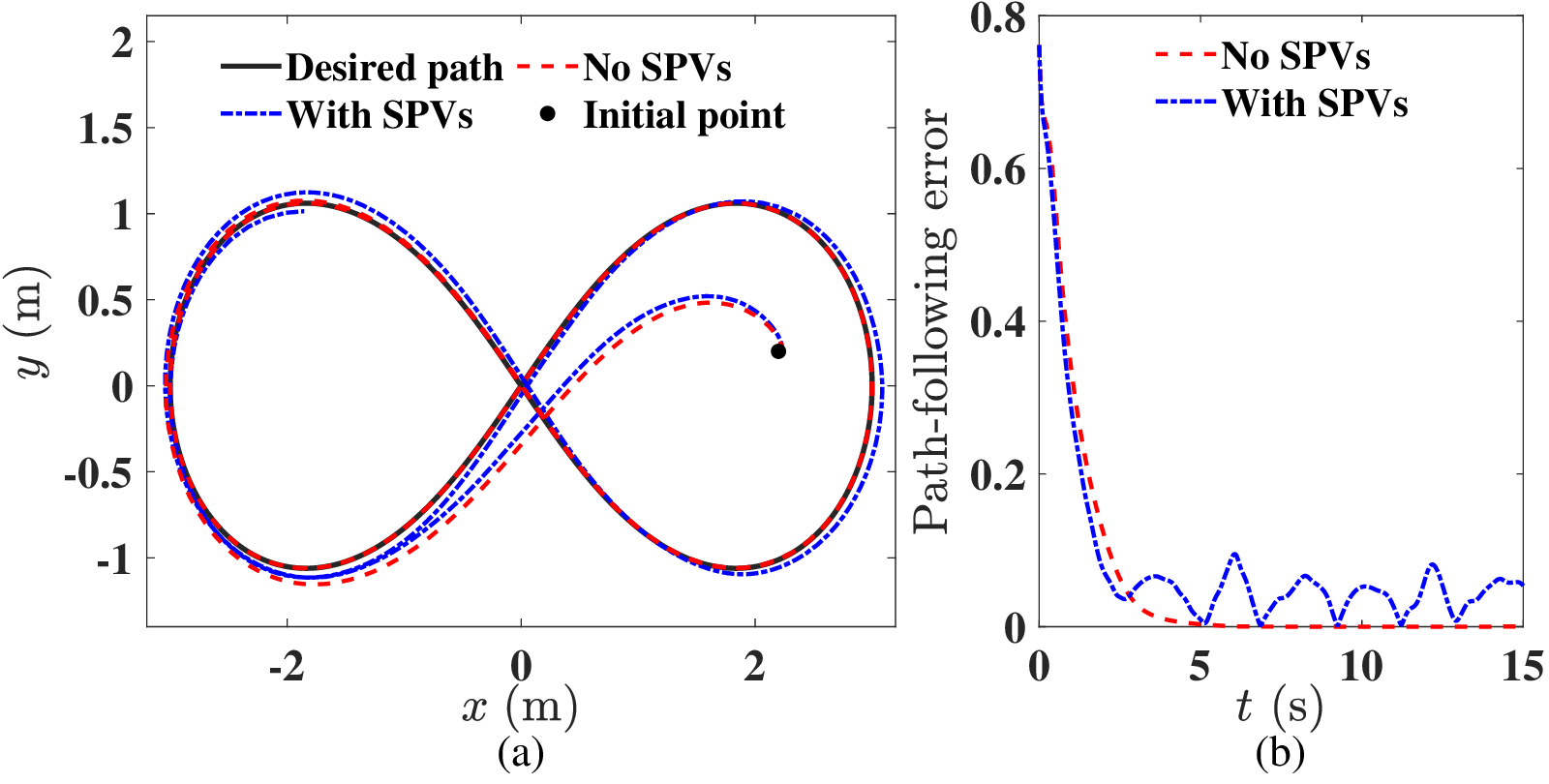}
            \caption{(a) Trajectories and (b) path-following errors of the nominal PVTOL aircraft system following path $\mathcal{P}_3$ under NVFCFC with system parameter variations (SPVs) of up to $10\%$, and without SPVs.}
            \label{nominal_robust}
\end{figure}

\begin{figure}[htbp]
  \centering
  \includegraphics[scale=0.3]{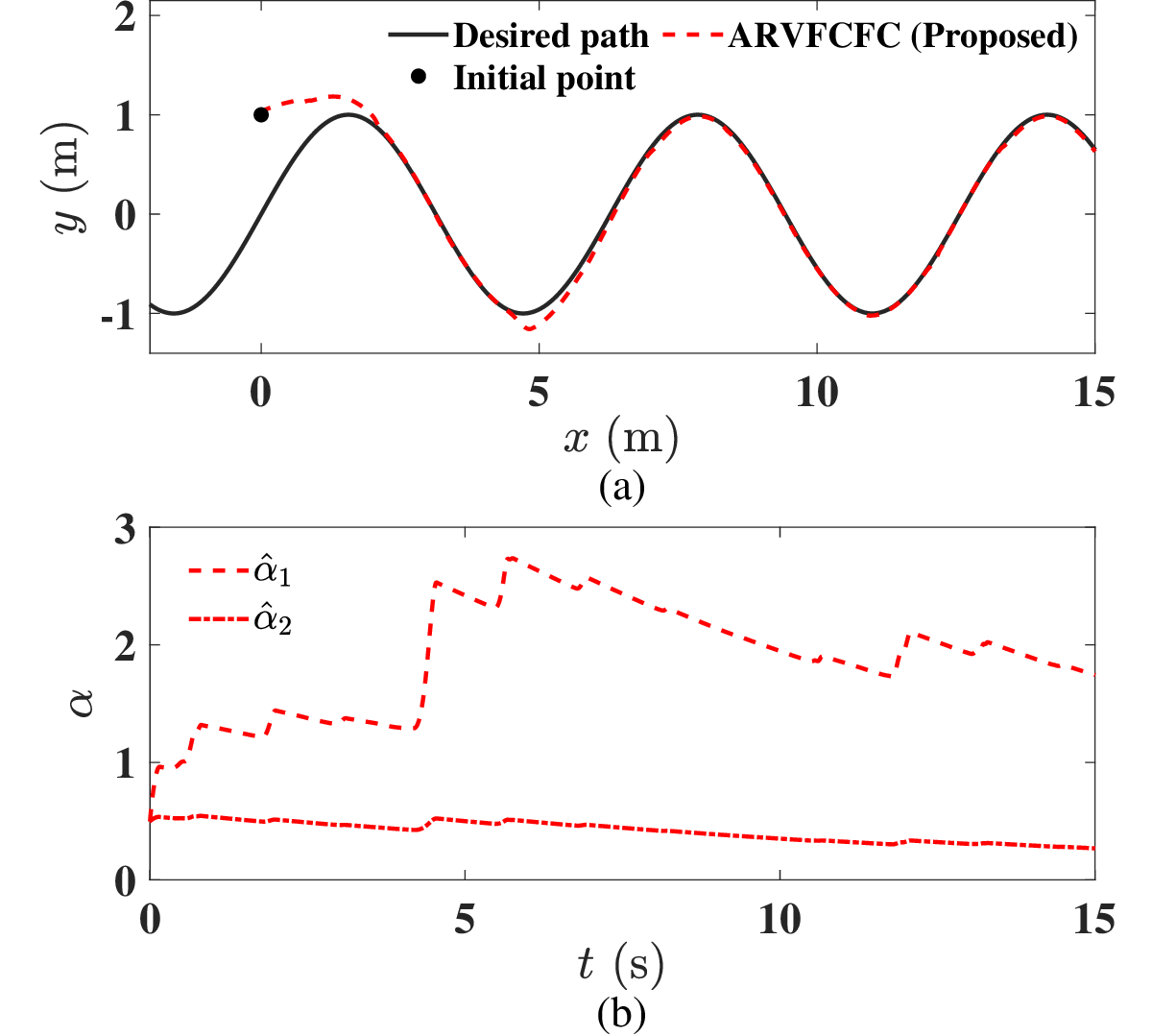}
  \caption{(a) Trajectory and (b) adaptive parameter histories of the uncertain PVTOL aircraft system following path $\mathcal{P}_1$ with ARVFCFC.} 
  \label{sin_adaptive}
\end{figure}

In the second step, the parameters of $p_3$ in ARVFCFC are chosen as $\ell_{1}=0.5$, $\ell_{2}=0.1$, $\varepsilon=1$. Other parameters are chosen to be the same as those in NVFCFC. NVFCFC and IIOS are also tested for comparisons. Figs. \ref{sin_adaptive}--\ref{p123_p12_path} show the corresponding simulation results. ARVFCFC completes all three path-following tasks with ultimately bounded path-following errors.  After some time adaptive parameters $\hat{\alpha}_{1,2}$ enter some positive region, consistent with Theorem \ref{thm2} (Figs. \ref{sin_adaptive}(b), \ref{adaptive_compare_path}(b), \ref{p123_p12_path}(b)). IIOS only completes the path-following task of $\mathcal{P}_2$, but produces a significantly larger path-following error after convergence compared with ARVFCFC, despite consuming significantly higher control effort  (Figs. \ref{adaptive_compare_path} and \ref{adaptive_compare_tau}). NVFCFC fails to achieve the path-following task of $\mathcal{P}_3$ since the system trajectory significantly deviates from $\mathcal{P}_3$ (Fig. \ref{p123_p12_path}). This reflects the effectiveness of incorporating $p_3$ to address uncertainties.
\begin{figure}[htbp]
  \centering
  \includegraphics[scale=0.3]{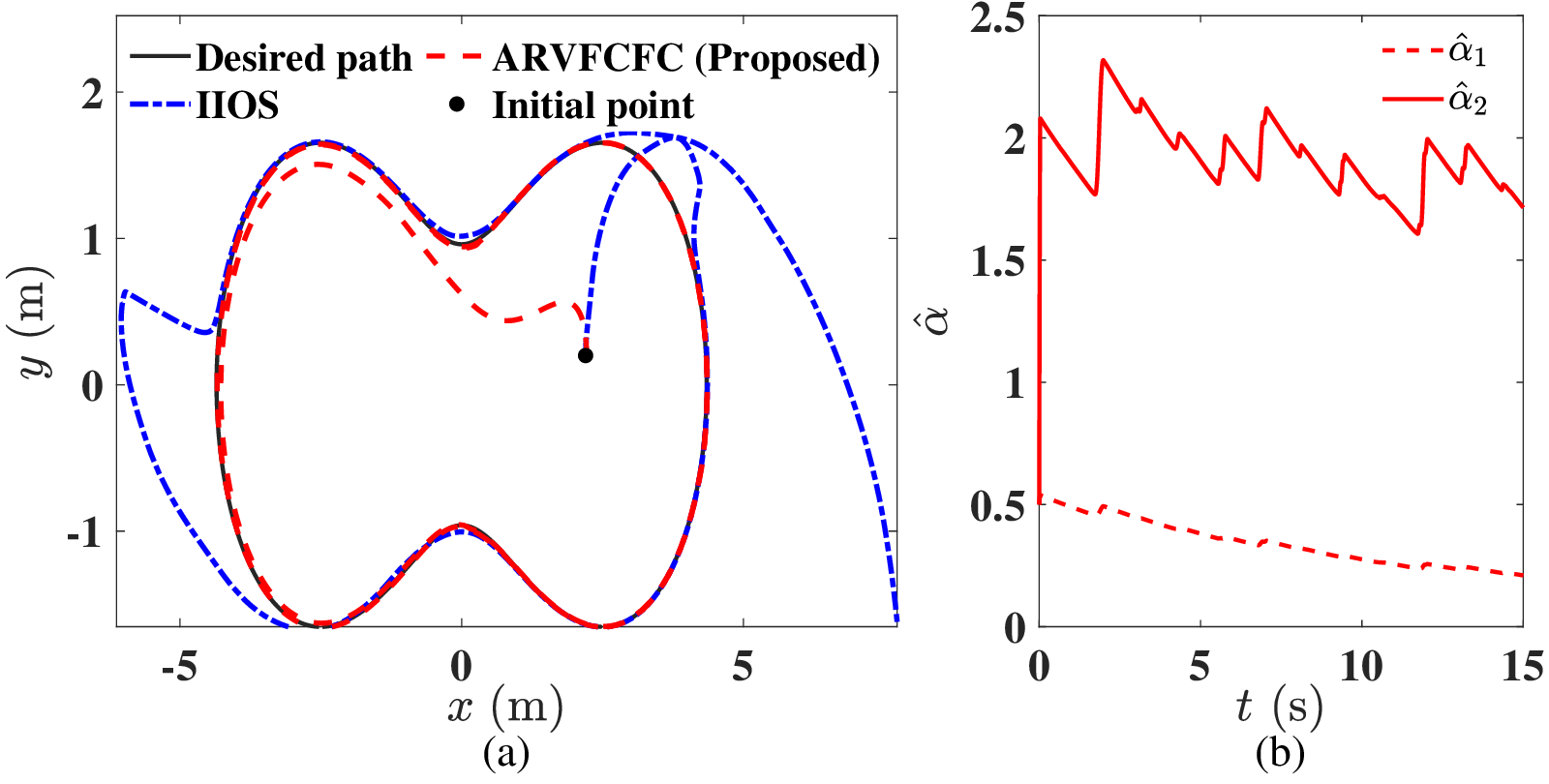}
  \caption{(a) Trajectories of the uncertain PVTOL aircraft system following path $\mathcal{P}_2$, and (b) corresponding adaptive parameter histories of ARVFCFC.} 
  \label{adaptive_compare_path}
\end{figure}
\begin{figure}[htbp]
            \centering
            \includegraphics[scale=0.3]{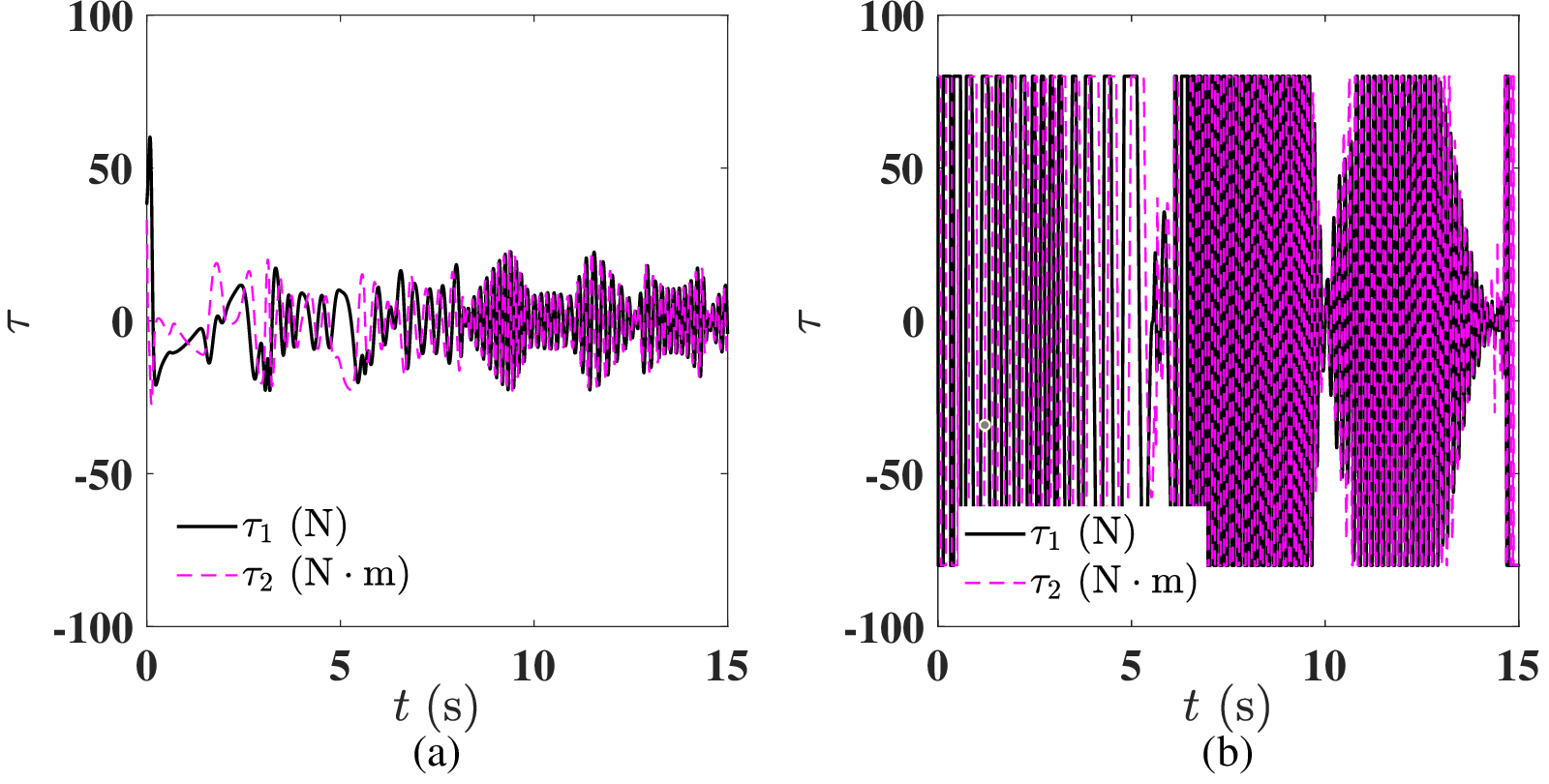}
            \caption{Control inputs of the uncertain PVTOL aircraft system following $\mathcal{P}_2$ with (a) ARVFCFC and (b) IIOS.}
            \label{adaptive_compare_tau}
\end{figure}
\begin{figure}[htbp]
  \centering
  \includegraphics[scale=0.3]{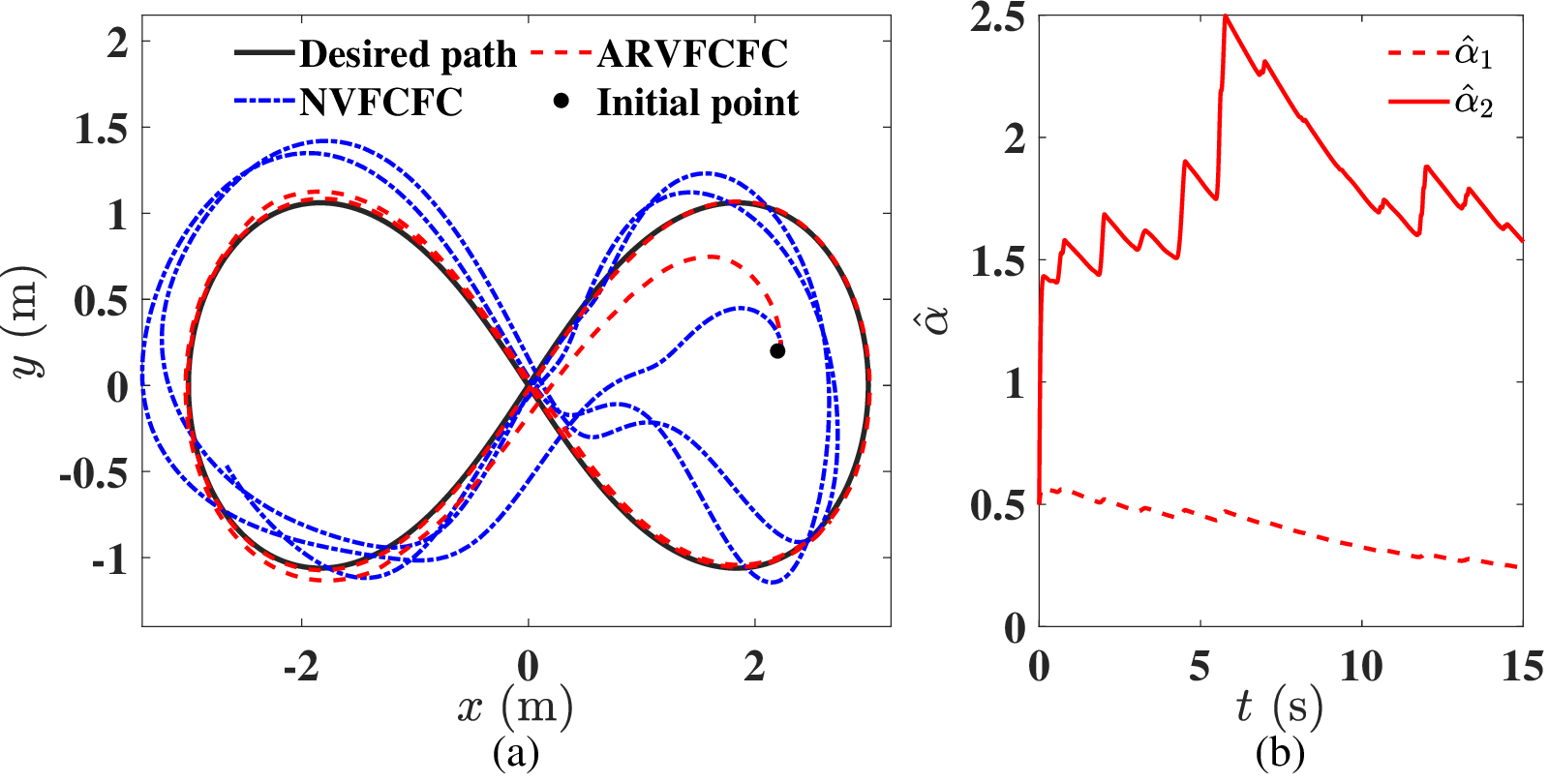}
  \caption{(a) Trajectories of the uncertain PVTOL aircraft system following path $\mathcal{P}_3$, and (b) corresponding adaptive parameter histories of ARVFCFC.} 
  \label{p123_p12_path}
\end{figure}
\subsection{A fully-actuated 3-link space manipulator system}
The dynamic model of a fully-actuated 3-link space manipulator \cite{shi2016robust} can be put in the
form of \eqref{mechanicalsystem} with
\begin{equation*}
\begin{split}
    &M=\begin{bmatrix}
      m_{12} &0 &0\\
      0 &m_{22} &m_{23}\\
      0 &m_{32} &m_{33}\\
      \end{bmatrix},
    C\dot{q}=\begin{bmatrix}
           C_{1}\\
           C_{2}\\
           C_{3}
      \end{bmatrix},
    g=\begin{bmatrix}
       d_{\theta_1}\\
       g_{2}+d_{\theta_2}\\
       g_{3}+d_{\theta_3}
      \end{bmatrix},\\
     &B=I_3,\quad
      \tau=[\tau_1,\tau_2,\tau_3]^\top,
      \end{split}
\end{equation*}
\begin{equation*}
    \begin{split}    &m_{12}=m_{3}l^{2}_{2}c_{23}^{2}+2m_{3}l_{1}l_{2}c_{2}c_{23}+m_{2}m_{3}l_{1}^{2}c_{2}^{2}+J,\\ &m_{22}=m_{3}l_{2}^{2}+2m_{3}l_{1}l_{2}c_{3}+(m_{2}+m_{3})l_{1}^{2},\\
    &m_{23}=m_{32}=m_{3}l_{2}(l_{1}+l_{2}c_{3}), m_{33}=m_{3}l_{2}^{2},\\ &C_{1}=-2m_{3}l_{2}^{2}\dot{\theta}_{1}\dot{\theta}_{2}s_{23}c_{23}-2m_{3}l_{2}^{2}\dot{\theta}_{1}\dot{\theta}_{3}s_{23}c_{23}\\
    &-2m_{3}l_{1}l_{2}\dot{\theta}_{1}\dot{\theta}_{2}s_{2}c_{23}-2m_{3}l_{1}l_{2}\dot{\theta}_{1}\dot{\theta}_{2}c_{2}s_{23}\\
    &-2m_{3}l_{1}l_{2}\dot{\theta}_{1}\dot{\theta}_{3}c_{2}s_{23}-2(m_{2}+m_{3})l_{1}^2\dot{\theta}_{1}\dot{\theta}_{2}s_{2}c_{2},\\
    &C_{2}=m_{3}l_{2}^2\dot{\theta}_{1}^2s_{23}c_{23}+m_{3}l_{1}l_{2}\dot{\theta}_{2}^2s_{3}+m_{3}l_{1}l_{2}\dot{\theta}_{2}^2c_{2}s_{23}\\
    &-m_{3}l_{1}l_{2}\dot{\theta}_{1}^2c_{23}^2s_{3}-m_{3}l_{1}l_{2}(\dot{\theta}_{2}+\dot{\theta}_{3})^2s_{3}\\
    &+m_{3}l_{1}l_{2}\dot{\theta}_{1}^2s_{23}c_{23}c_{3}+(m_{2}+m_{3})l_{1}^2\dot{\theta}_{1}^2s_{2}c_{2},\\
    &C_{3}=m_{3}l_{2}^2\dot{\theta}_{1}^2s_{23}c_{23}+m_{3}l_{1}l_{2}\dot{\theta}_{2}^2s_{3}+m_{3}l_{1}l_{2}\dot{\theta}_{1}^2c_{2}s_{23},\\
    &g_{2}=m_{3}g_0l_{2}c_{23}+(m_{2}+m_{3})gl_{1}c_{2}, 
    g_{3}=m_{3}g_0l_{2}c_{23},
    \end{split}
\end{equation*}
and $q=[\theta_1,\theta_2,\theta_3]^{\top}$ are the angles of links of the manipulator, $\tau=[\tau_{1},\tau_{2},\tau_{3}]^{\top}$ respectively represent the control torture imposed on the links, $l_1$ is the length of link 2, $l_2$ is the length of link 3, $J$ is the the rotary inertia of link 1, $m_{2}$, $m_{3}$ are the masses of link 2 and 3, respectively, $s_{i}=\sin \theta_{i}$, $c_{i}=\cos \theta_{i}$, $i=2,3$, $s_{23}=\sin(\theta_{2}+\theta_{3})$, $c_{23}=\cos(\theta_{2}+\theta_{3})$, $g_0$ is the gravitational acceleration. Furthermore, $d_{\theta_1}$, $d_{\theta_2}$, and $d_{\theta_3}$ are unknown external disturbances, and $m_i(i=2,3)$ and $J$ are considered to contain uncertainties, which can be decomposed as
$m_i=\bar{m}_i+\Delta m_i$, $J=\bar{J}+\Delta J$.
These uncertainties are all bounded with unknown bounds.

Denote the coordinates of the end effector of the manipulator as $[x,y,z]$, for which we choose a self-intersecting desired path $\mathcal{P}_4$ that is the intersection of two cylinders described by $(x-x^c)^{2}+(z-z^c)^{2}=R_1^{2}$ and $(y-y^c)^{2}+z^{2}=R_2^{2}$. We can find a parametrization for the desired path as follows:
\begin{equation*}
    \begin{split}
        &x=f_{1}(w)=x^c+R_1\sin(2w),\\
        &y=f_{2}(w)=y^c+2\sin w \sqrt{R_1(R_2-R_1\sin ^{2}w)},\\
        &z=f_{3}(w)=z^c+R_1\cos(2w).
    \end{split}
    \label{desiredpath_3l}
\end{equation*}
We choose $R_2 = 1.5, R_1 = 0.5, x^c = 0, y^c = 2.5, z^c=R_2-R_1=1$. To further evaluate the proposed controller, we consider another desired path $\mathcal{P}_5$ given by a non-self-intersecting $(3,2)$-torus knot, which is parameterized by
\begin{equation*}
\begin{split}
    &x=f_1(\omega)  = (1 + 0.2\cos(3\omega))\cos(2\omega),\\
    &y=f_2(\omega)  = 2+ (1 + 0.2\cos(3\omega))\sin(2\omega),\\
    &z=f_3(\omega)  = 1+0.2\sin(3\omega).
    \end{split}
\end{equation*} 
 
 Since the path-following task is directly assigned to the end effector of the manipulator, while the system model is established by the link angles $\theta_1$, $\theta_2$ and $\theta_3$, we need to invoke the analytical expressions that relate the angles $[\theta_1,\theta_2,\theta_3]$ to the end effector coordinates $[x,y,z]$, which are:
\begin{equation}
    \begin{split}
        &\theta_1=\arctan\frac{x}{y},\\
        &\theta_2=\arccos\frac{x^2+y^2+z^2+l_1^2-l_2^2}{2l_1\sqrt{x^2+y^2+z^2}}+\arccos\sqrt{\frac{x^2+y^2}{x^2+y^2+z^2}},\\
        &\theta_3=\arccos\frac{x^2+y^2+z^2-l_1^2-l_2^2}{2l_1l_2}.\\
    \end{split}
    \label{x2theta}
\end{equation}
Consequently, the corresponding VFCs for $\mathcal{P}_4$ and $\mathcal{P}_5$ can be obtained by \eqref{VFC}, with $A=I_{3\times 3}$ for both paths, and the particular expressions of $\chi^{\rm s}$ are omitted here.

Similar to assumptions verification in subsection \ref{underactuated}, Assumptions \ref{feasbility}--\ref{assump_tau2} can be verified straightforwardly. Assumption \ref{assump5} can be verified by $W=\bar{M}M^{-1}$ indicating $\frac{1}{2} \min_{\sigma\in\Sigma} \lambda_m(W+W^{\top})> 0$. Regarding the inequality \eqref{assumption6} in Assumption \ref{assump6}, note that the terms in $M$, $C\dot{q}$, and $g$ are either constant, bounded by a constant, trigonometric in positions, or quadratic in velocities; besides, there are additional terms corresponding to the VFCs that are constant, or trigonometric in $\omega$. Therefore, to meet Assumption \ref{assump6}, we choose 
\begin{equation}
    \Pi(\alpha,q,\dot{q},w,t)=\alpha(\|\dot{q}\|^{2}+\|\dot{q}\|+1)=:\alpha \breve{\Pi}(q,\dot{q},w,t),
\end{equation}
where $\alpha\in \mathbb{R}_+$. 
For simulations, we choose $\bar{J}=0.5\,\rm{kg \cdot m^{2}}$, 
$\bar{m}_{2}=1\,\rm{kg}$, $\bar{m}_{3}=2\,\rm{kg}$ and the initial conditions are $\theta_1(0)=-0.1\,\rm{rad} $, $\theta_2(0)=1.5 \,\rm{rad}$, $\theta_3(0)=1.5\,\rm{rad}$, $\dot{\theta}_1(0)=1\,\rm{rad/s}$, 
$\dot{\theta}_2(0)=0$, $\dot{\theta}_3(0)=0$, $\alpha(0)=0.5$, $\omega(0)=0$. The uncertainties are simulated by $\Delta J=0.1\bar{J}\sin(7.5t)$, $\Delta m_{2}=0.1\bar{m}_{2}\sin(5t)$, $\Delta m_{3}=0.1\bar{m}_{3}\cos(5t)$, $d_{\theta_1}=3\sin 2t$, $d_{\theta_2}=3\cos 4t$, 
$d_{\theta_3}=3\sin 6t$. The VFC and control parameters are chosen as $k_1=k_2=k_3=1$, $P=I$, $\kappa=5$, $\ell_{1}=0.5$, $\ell_{2}=0.3$, $\varepsilon=1$. 

The simulation results are presented in Figs. \ref{LRA_knot} and \ref{LRA_singularity_free}. 
ARVFCFC completes both path-following tasks with ultimately bounded path-following errors. In each task, the adaptive parameter enters some positive region after some time.

\begin{figure}[htbp]
            \centering
            \includegraphics[scale=0.34]{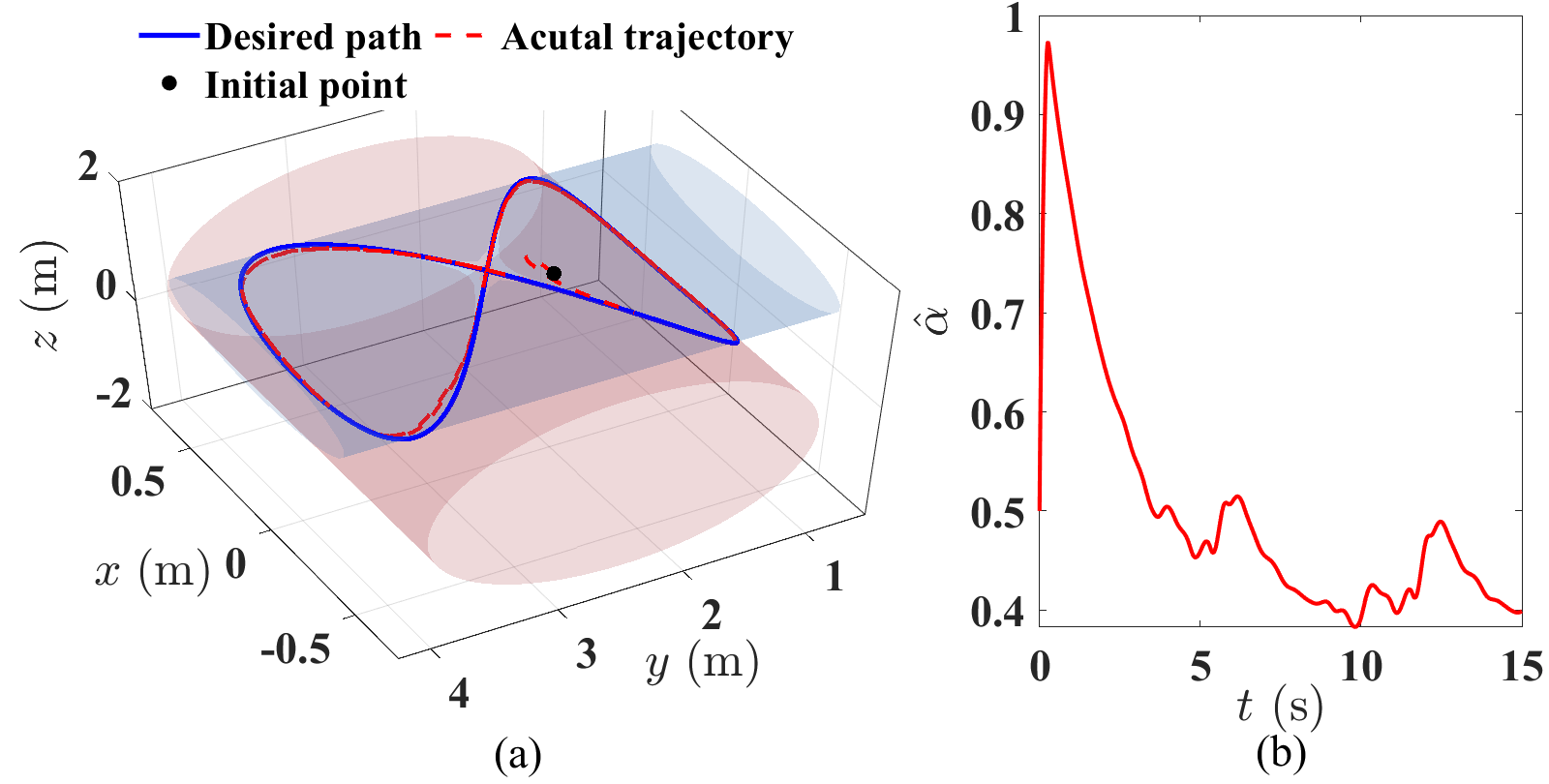}
            \caption{(a) Trajectory and (b) adaptive parameter history of the uncertain space manipulator system following self-intersecting path $\mathcal{P}_4$ with ARVFCFC.}
            \label{LRA_knot}
\end{figure}

\begin{figure}[htbp]
            \centering
            \includegraphics[scale=0.34]{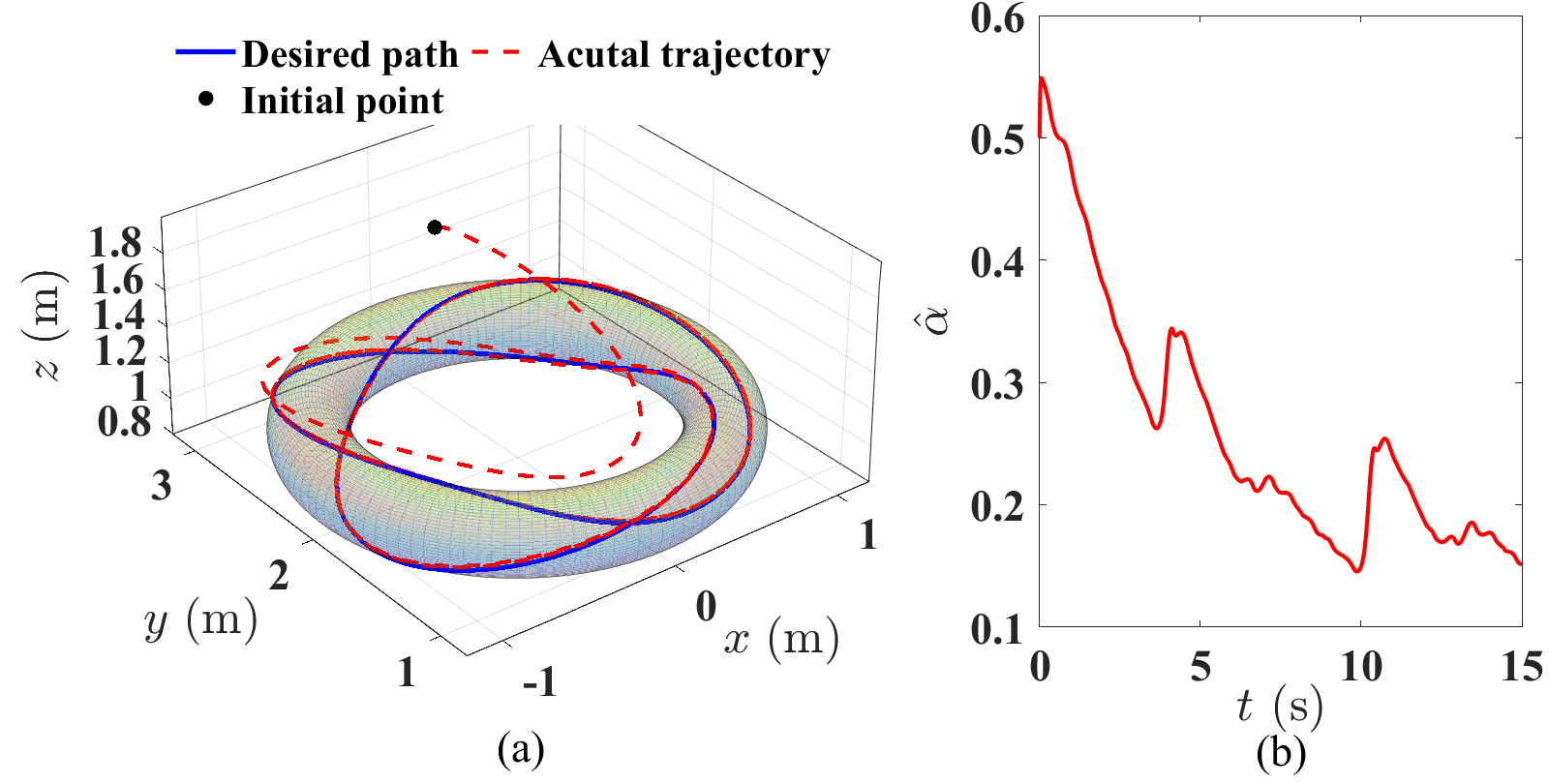}
            \caption{(a) Trajectory and (b) adaptive parameter history of the uncertain space manipulator system following non-self-intersecting path $\mathcal{P}_5$ with ARVFCFC.}
            \label{LRA_singularity_free}
\end{figure}

\section{Conclusion}  \label{sec_conclusion}
This article integrates GVF and CFC to solve the dynamics control problem of geometric path-following for mechanical systems under heterogeneous time-varying uncertainties--a task impossible with either approach alone. The VFC is meticulously constructed to encode the desired path, bridging GVF and CFC. Under mild assumptions, the main results include: i) path-following error convergence matches that of the VFC-following error; ii) nominal VFCFC ensures vanishing errors without uncertainties; iii) adaptive robust VFCFC ensures ultimately bounded errors with uncertainties. Simulation results on two examples validate the proposed control framework.

This work opens a new avenue for path-following control that considers the internal dynamics of uncertain systems. Future research may explore extensions such as collision avoidance, input saturation, and time-varying desired paths in swarm mechanical systems,  for enhanced robotic deployment.

\bibliographystyle{IEEEtran}
\bibliography{ref}
\appendix
\subsection*{Proof of Theorem \ref{thm2}}
	Let $V(\beta,\hat{\alpha}-\alpha)=\beta^{\top}P\beta+\ell_1^{-1}(1+\rho_{W})(\hat{\alpha}-\alpha)^{\top}(\hat{\alpha}-\alpha)$
	be the Lyapunov function candidate. The time derivative of $V$ is given by
	\begin{equation}\label{DeriLya}
  \setlength{\abovedisplayskip}{3pt}
\setlength{\belowdisplayskip}{3pt}
		\dot{V}=2\beta^{\top}P\dot{\beta}+2\ell_1^{-1}(1+\rho_{W})(\hat{\alpha}-\alpha)^{\top}\dot{\hat{\alpha}}.
	\end{equation}
	In the proof, the arguments of functions are mostly omitted when no confusion is likely to arise, except for a few critical cases. Now we deal with \eqref{DeriLya} by analyzing each term separately. First, with \eqref{eq:constraint-error-dynamics}, we have
	\begin{equation}\label{proof_beta_0}
  \setlength{\abovedisplayskip}{3pt}
\setlength{\belowdisplayskip}{3pt}
		\begin{split}
			2\beta^{\top}P\dot{\beta}=2\beta^{\top}P\{A[M^{-1}(B\tau-C\dot{q}-g)]-b\}.
		\end{split}
	\end{equation}
	By the decomposition for nominal and uncertain portions of system matrices and $M^{-1}=\bar{M}^{-1}+(M^{-1}-\bar{M}^{-1})$, we obtain
	\begin{equation}\label{proof_beta_1}
  \setlength{\abovedisplayskip}{3pt}
\setlength{\belowdisplayskip}{3pt}
		\begin{split}
			&\beta^{\top}P\{A[M^{-1}(B\tau-C\dot{q}-g)]-b\}\\
			=&\beta^{\top}P\{A[\bar{M}^{-1}(-\bar{C}\dot{q}-\bar{g})+\bar{M}^{-1}B(p_{1}+p_{2})+\bar{M}^{-1}(-\Delta C\dot{q}\\
			&-\Delta g)+(M^{-1}-\bar{M}^{-1})(-C\dot{q}-g+Bp_{1}+Bp_{2})\\
			&+M^{-1}Bp_{3}]-b\}.
		\end{split}
	\end{equation}
	By $B=\bar{B}+\Delta B$, $\Delta B= \bar{B}\check{B}+\tilde{B}$ and $p_1$ in \eqref{nom_p1} we have
	\begin{equation}\label{eq33}
  \setlength{\abovedisplayskip}{3pt}
\setlength{\belowdisplayskip}{3pt}
		\begin{split}
			A&[\bar{M}^{-1}(-\bar{C}\dot{q}-\bar{g})+\bar{M}^{-1}Bp_{1}]-b\\
			=&A[\bar{M}^{-1}(-\bar{C}\dot{q}-\bar{g})+\bar{M}^{-1}(\bar{{B}}+\bar{B}\check{B}+\tilde{B})p_{1}]-b\\
			=&A[\bar{M}^{-1}(-\bar{C}\dot{q}-\bar{g})+\bar{M}^{-1}\bar{B}p_{1}]-b+A\bar{M}^{-1}\bar{B}\check{B}p_{1}\\   &+A\bar{M}^{-1}\tilde{B}p_{1}=A\bar{M}^{-1}\bar{B}\check{B}p_{1}.
		\end{split}
	\end{equation}
Using \eqref{eq33} and the identities as \eqref{uncertainty_identity} deduced from match-mismatch uncertainty decompositions, the terms on the RHS of \eqref{proof_beta_1} excluding the $p_3$ term can be simplified as
\begin{equation}
\begin{aligned}
& AM^{-1}\bigl[-C\dot q-g+B(p_1+p_2)\bigr]-b\\
&=A\bar M^{-1}\bigl[
  \Delta Bp_1-\Delta C\dot q-\Delta g+Bp_2\\
&\qquad+H(-C\dot q-g+Bp_1+Bp_2)\bigr]\\
&=A\bar M^{-1}\bar B\bigl[
  \check H(-C\dot q-g+Bp_1)+\check Bp_1-\check C\dot q-\check g\bigr]\\
&\quad+A\bar M^{-1}\bar B\check H Bp_2
  +A\bar M^{-1}Bp_2.
\end{aligned}
\label{proof_beta_2}
\end{equation}

	The last two terms on the RHS of \eqref{proof_beta_2} can be rewritten as:
	\begin{equation}\label{proof_beta_3}
  \setlength{\abovedisplayskip}{2pt}
\setlength{\belowdisplayskip}{2pt}
		\begin{split}
		A&\bar{M}^{-1}\bar{B}\check{H}Bp_{2}+A\bar{M}^{-1}Bp_{2}\\
			=&A\bar{M}^{-1}(H-\tilde{H})(\bar{B}+\Delta B)p_{2}+A\bar{M}^{-1}(\bar{B}+\Delta B)p_{2}\\
			=&A\{M^{-1}+\bar{M}^{-1}[H-\bar{B}(A\bar{M}^{-1}\bar{B})^{-1}A\bar{M}^{-1}H]\}(\bar{B}\\
   &+\Delta B)p_{2}
			=AM^{-1}(\bar{B}+\Delta B)p_{2}.
		\end{split}
	\end{equation}
	By (\ref{proof_beta_0})-(\ref{proof_beta_3}), we can obtain
	\begin{equation}\label{proof_beta_9}
  \setlength{\abovedisplayskip}{2pt}
\setlength{\belowdisplayskip}{2pt}
		\begin{split}
			2&\beta^{\top}P\dot{\beta}=2\beta^{\top}PA\bar{M}^{-1}\bar{B}[\check{H}(-C\dot{q}-g+Bp_{1})+(\check{B}p_{1}\\
			&-\check{C}\dot{q}-\check{g})]+2\beta^{\top}PAM^{-1}(\bar{B}+\Delta B)p_{2}\\
			&+2\beta^{\top}PAM^{-1}(\bar{B}+\Delta B)p_{3}.\\
		\end{split}
	\end{equation}
	Considering the first two terms on the RHS of \eqref{proof_beta_9}, by Assumption \ref{assump6}, we have
	\begin{equation*}
  \setlength{\abovedisplayskip}{2pt}
\setlength{\belowdisplayskip}{2pt}
		\begin{split}
			2&\beta^{\top}PA\bar{M}^{-1}\bar{B}[\check{H}(-C\dot{q}-g+Bp_{1})+(\check{B}p_{1}\\
			&-\check{C}\dot{q}-\check{g})]\le 2\|\breve{\beta}\|(1+\rho_{W})\Pi(\alpha,q,\dot{q},w,t).
		\end{split}
	\end{equation*}
	Considering the last two terms on the RHS of \eqref{proof_beta_9}, we have
	\begin{equation*}
  \setlength{\abovedisplayskip}{0pt}
\setlength{\belowdisplayskip}{0pt}
		\begin{split}
			2&\beta^{\top}PAM^{-1}(\bar{B}+\Delta B)p_{2}+2\beta^{\top}PAM^{-1}(\bar{B}+\Delta B)p_{3}\\
			=&2\beta^{\top}PAM^{-1}(\bar{B}+\Delta B)(p_{2}+p_{3}),
		\end{split}
	\end{equation*}
	where
	\begin{equation*}
  \setlength{\abovedisplayskip}{2pt}
\setlength{\belowdisplayskip}{2pt}
		\begin{split}
			2&\beta^{\top}PAM^{-1}(\bar{B}+\Delta B)\\
			=&2\beta^{\top}PA\bar{M}^{-1}[I+(\bar{B}\check{H}+\tilde{H})](\bar{B}+\bar{B}\check{B}+\tilde{B})\\
			=&2\beta^{\top}PA\bar{M}^{-1}\bar{B}[I+\check{B}+\check{H}(\bar{B}+\bar{B}\check{B}+\tilde{B})]\\
			=&2\beta^{\top}PA\bar{M}^{-1}\bar{B}(I+\check{B}+\check{H}B).
		\end{split}
	\end{equation*}
	By $p_2$ in \eqref{nominal}, we have
	\begin{equation*}
  \setlength{\abovedisplayskip}{0pt}
\setlength{\belowdisplayskip}{0pt}
		\begin{split}
			2&\beta^{\top}PAM^{-1}(\bar{B}+\Delta B)p_{2}\\
			=&2\beta^{\top}PA\bar{M}^{-1}\bar{B}(I+\check{B}+\check{H}B)(-\kappa (A\bar{M}^{-1} \bar{B})^{\top} P\beta)\\
			=&2\beta^{\top}PA\bar{M}^{-1}\bar{B}(-\kappa \bar{B}^{\top}\bar{M}^{-1} A^{\top} P \beta)+2\beta^{\top}PA\bar{M}^{-1}\bar{B}\\
			&\times (\check{B}+\check{H}B)(-\kappa \bar{B}^{\top}\bar{M}^{-1} A^{\top} P \beta).
		\end{split}
	\end{equation*}
	Applying Assumptions \ref{assump_tau2} and \ref{assump5} yields
	\begin{equation*}
  \setlength{\abovedisplayskip}{2pt}
\setlength{\belowdisplayskip}{2pt}
		\begin{split}
			&2\beta^{\top}PAM^{-1}(\bar{B}+\Delta B)p_{2}=-2\kappa\beta^{\top}(PA\bar{M}^{-1}\bar{B}\bar{B}^{\top}\\
			&\times\bar{M}^{-1} A^{\top} P )\beta-2\beta^{\top}PA\bar{M}^{-1}\bar{B}\frac{1}{2}(W+W^{\top})\\
   &\times(\kappa B^{\top}\bar{M}^{-1} A^{\top} P \beta)\le  -2\kappa(1+\rho_{W})\underline{\lambda}\|\beta\|^{2}.
		\end{split}
	\end{equation*}
	Similarly, by $p_3$ in (\ref{CFC_p3}) we have
	\begin{equation}\label{proof_beta_8}
  \setlength{\abovedisplayskip}{0pt}
\setlength{\belowdisplayskip}{0pt}
		\begin{split}
			2&\beta^{\top}PAM^{-1}(\bar{B}+\Delta B)p_{3}\\
			=&2\beta^{\top}PA\bar{M}^{-1}\bar{B}(-\eta\upsilon\Pi(\alpha,q,\dot{q},w,t))+2\beta^{\top}PA\bar{M}^{-1}\\
   &\times\bar{B}(\check{B}+\check{H}B)(-\eta\upsilon\Pi(\hat{\alpha},q,\dot{q},w,t))\\
			=&2\beta^{\top}PA\bar{M}^{-1}\bar{B}(-\eta\bar{B}^{\top}\bar{M}^{-1}A^{\top}P\beta\Pi^{2}(\hat{\alpha},q,\dot{q},w,))\\
			&+2\beta^{\top}PA\bar{M}^{-1}\bar{B}(\check{B}+\check{H}B)(-\eta\bar{B}^{\top}\bar{M}^{-1}A^{\top}P\beta\\
   &\times\Pi^{2}(\hat{\alpha},q,\dot{q},w,t)).
		\end{split}
	\end{equation}
	By the expression of $\upsilon$ in \eqref{CFC_p3}, the first term on the RHS of \eqref{proof_beta_8} is 
	\begin{equation}\label{proof_beta_8-1}
  \setlength{\abovedisplayskip}{0pt}
\setlength{\belowdisplayskip}{0pt}
		\begin{split}
			2&\beta^{\top}PA\bar{M}^{-1}\bar{B}(-\eta\bar{B}^{\top}\bar{M}^{-1}A^{\top}P\beta\Pi^{2}(\hat{\alpha},q,\dot{q},w,t))\\
			=&-2\eta\Pi(\hat{\alpha},q,\dot{q},w,t)\beta^{\top}PA\bar{M}^{-1}\bar{B}\bar{B}^{\top}\bar{M}^{-1}A^{\top}P\beta\\
   &\times\Pi(\hat{\alpha},q,\dot{q},w,t)\\
			=&-2\eta\|\upsilon(\hat{\alpha},q,\dot{q},w,t)\|^{2}.
		\end{split}
	\end{equation}
	Based on Assumption \ref{assump5}, the second term on the RHS of \eqref{proof_beta_8} is 
	\begin{equation}\label{proof_beta_8-2}
  \setlength{\abovedisplayskip}{0pt}
\setlength{\belowdisplayskip}{0pt}
		\begin{split}
			2&\beta^{\top}PA\bar{M}^{-1}\bar{B}(\check{B}+\check{H}B)(-\eta\bar{B}^{\top}\bar{M}^{-1}A^{\top}P\beta\\
   &\times\Pi^{2}(\hat{\alpha},q,\dot{q},w,t))\\
			=&-2\eta(\Pi(\hat{\alpha},q,\dot{q},w,t)\beta^{\top}PA\bar{M}^{-1}\bar{B}\frac{1}{2}(W+W^{\top})(\bar{B}^{\top}\\
   &\times\bar{M}^{-1}A^{\top}P\beta\Pi(\hat{\alpha},q,\dot{q},w,t))\\
   &\le -2\eta\rho_{W}\|\upsilon(\hat{\alpha},q,\dot{q},w,t)\|^{2}.
		\end{split}
	\end{equation}

 Combining \eqref{proof_beta_8}-\eqref{proof_beta_8-2} we can obtain
	\begin{equation*}
  \setlength{\abovedisplayskip}{0pt}
\setlength{\belowdisplayskip}{0pt}
		2\beta^{\top}PAM^{-1}(\bar{B}+\Delta B)p_{3}\le -2\eta(1+\rho_{W})\|\upsilon(\hat{\alpha},q,\dot{q},w,t)\|^{2}.
	\end{equation*}

	Based on the expression of $\eta$ in \eqref{CFC_p3}, if $\|\upsilon\|>\mu$, $-2\eta(1+\rho_{W})\|\upsilon\|^{2}=-2(1+\rho_{W})\|\upsilon\|$ and 
	\begin{equation}\label{proof_beta_4}
  \setlength{\abovedisplayskip}{0pt}
\setlength{\belowdisplayskip}{0pt}
		\begin{split}
			2&\beta^{\top}P\dot{\beta}\le 2\|\breve{\beta}\|(1+\rho_{W})\Pi(\alpha,q,\dot{q},w,t)\\
    &-2\kappa(1+\rho_{W})\underline{\lambda}\|\beta\|^{2}-2(1+\rho_{W})\|\upsilon(\hat{\alpha},q,\dot{q},w,t)\|\\
    =&2\|\breve{\beta}\|(1+\rho_{W})\Pi(\alpha,q,\dot{q},w,t)-2\kappa(1+\rho_{W})\underline{\lambda}\|\beta\|^{2}\\
    &-2\|\breve{\beta}\|(1+\rho_{W})\Pi(\hat{\alpha},q,\dot{q},w,t)\\
			=&-2\kappa(1+\rho_{W})\underline{\lambda}\|\beta\|^{2}-2\|\breve{\beta}\|(1+\rho_{W})(\hat{\alpha}-\alpha)^{\top}\breve{\Pi}(q,\dot{q},w,t).
		\end{split}
	\end{equation}
	And if $\|\upsilon\|\le \mu$, $-2\eta(1+\rho_{W})\|\upsilon\|^{2}=-2(1+\rho_{W})\|\upsilon\|^{2}/\mu$ and
	\begin{equation}\label{proof_beta_5}
		\begin{split}
			&2\beta^{\top}P\dot{\beta}\\
            &\le 2\|\breve{\beta}\|(1+\rho_{W})\Pi(\alpha,q,\dot{q},w,t)-2\kappa(1+\rho_{W})\underline{\lambda}\|\beta\|^{2}\\
		    &-2(1+\rho_{W})\frac{\|\upsilon(\hat{\alpha},q,\dot{q},w,t)\|^{2}}{\mu}\\
			&= 2(1+\rho_{W})\|\upsilon(\hat{\alpha},q,\dot{q},w,t)\|(1-\frac{\|\upsilon(\hat{\alpha},q,\dot{q},w,t)\|}{\mu})\\
   &-2\kappa(1+\rho_{W})\underline{\lambda}\|\beta\|^{2}-2\|\breve{\beta}\|(1+\rho_{W})(\hat{\alpha}-\alpha)^{\top}\breve{\Pi}(q,\dot{q},w,t)\\
   &\le \frac{1}{2}(1+\rho_{W})\mu-2\kappa(1+\rho_{W})\underline{\lambda}\|\beta\|^{2}\\
   &-2\|\breve{\beta}\|(1+\rho_{W})(\hat{\alpha}-\alpha)^{\top}\breve{\Pi}(q,\dot{q},w,t).
		\end{split}
	\end{equation}
	Therefore, using \eqref{proof_beta_4} and \eqref{proof_beta_5}, there is 
	\begin{equation}\label{proof_beta_6}
		\begin{split}
			2&\beta^\top P\dot{\beta}\le \frac{1}{2}(1+\rho_{W})\mu-2\kappa(1+\rho_{W})\underline{\lambda}\|\beta\|^{2}\\
			&-2\|\breve{\beta}\|(1+\rho_{W})(\hat{\alpha}-\alpha)^{\top}\breve{\Pi}(q,\dot{q},w,t).
		\end{split}
	\end{equation}
	for all $\upsilon$.
	Subsequently, we will investigate the second term of \eqref{proof_beta_0}. Based on the adaptive law \eqref{CFC_zi2}, if $\|\breve{\Pi}\|\|\breve{\beta}\|>\varepsilon$, we have
	\begin{equation}\label{proof_adaptive_1}
		\begin{split}
			&2(1+\rho_{W})(\hat{\alpha}-\alpha)^{\top}\ell_1^{-1}\dot{\hat{\alpha}}\\
			&=2(1+\rho_{W})(\hat{\alpha}-\alpha)^{\top}\breve{\Pi}(q,\dot{q},w,t)\|\breve{\beta}\|\\
   &-2(1+\rho_{W})(\hat{\alpha}-\alpha)^{\top}\ell_1^{-1}\ell_2\hat{\alpha}\\
			&\le 2(1+\rho_{W})(\hat{\alpha}-\alpha)^{\top}\breve{\Pi}(q,\dot{q},w,t)\|\breve{\beta}\|\\
   &-2(1+\rho_{W})\|\hat{\alpha}-\alpha\|^{2} \ell_1^{-1}\ell_2+2(1+\rho_{W})\|\hat{\alpha}-\alpha\| \|\alpha \|\ell_1^{-1}\ell_2;
		\end{split}
	\end{equation}
 if $\|\breve{\Pi}\|\|\breve{\beta}\|\le\varepsilon$, we have
    \begin{equation}\label{proof_adaptive_2}
		\begin{split}
			&2(1+\rho_{W})(\hat{\alpha}-\alpha)^{\top}\ell_1^{-1}\dot{\hat{\alpha}}\\
			&=2(1+\rho_{W})(\hat{\alpha}-\alpha)^{\top}\breve{\Pi}(q,\dot{q},w,t)\frac{\|\breve{\Pi}(q,\dot{q},w,t)\|\|\breve{\beta}\|^2}{\varepsilon}\\
   &-2(1+\rho_{W})(\hat{\alpha}-\alpha)^{\top}\ell_1^{-1}\ell_2\hat{\alpha}.
		\end{split}
	\end{equation}
	Combining \eqref{proof_beta_0}, \eqref{proof_beta_6} and \eqref{proof_adaptive_1}, we can obtain
	\begin{equation}\label{proof_beta_10}
		\begin{split}
			\dot{V}&\le \frac{1}{2}(1+\rho_{W})\mu-2\kappa(1+\rho_{W})\underline{\lambda}\|\beta\|^{2}-2(1+\rho_{W})\\
			&\times\|\hat{\alpha}-\alpha\|^{2}\ell_1^{-1}\ell_2+2(1+\rho_{W})\|\hat{\alpha}-\alpha\| \|\alpha \|\ell_1^{-1}\ell_2.
		\end{split}
	\end{equation}
    Combining \eqref{proof_beta_0}, \eqref{proof_beta_6} and \eqref{proof_adaptive_2}, we can obtain a result similar to \eqref{proof_beta_10} . Let $\varphi(t):=[\beta^{\top}(q(t),\dot{q}(t),w(t))\quad(\hat{\alpha}(t)-\alpha)^{\top} ]^{\top}$, then we have 
	\begin{equation*}
		\dot{V}\le -K_{1}\|\varphi\|^{2}+K_{2}\|\varphi\|+K_{3},
	\end{equation*}
	where $K_{1}=\min\{2\kappa \underline{\lambda}(1+\rho_{W}), 2(1+\rho_{W})\ell_1^{-1}\ell_2\}$, $K_{2}=(1+\rho_{W})(2\ell_1^{-1}\ell_2\|\alpha\|+\frac{\varepsilon}{2})$, $K_{3}=\frac{\mu}{2}(1+\rho_{W})$.
	Based on the standard arguments on UB and UUB in \cite{corless1981continuous}, it can be concluded that $\|\varphi\|$ is uniformly bounded by
	\begin{equation*}
		d(s)=\begin{cases}
			\sqrt{\frac{\mathcal{X}_{2}}{\mathcal{X}_{1}}}R, & \mbox{if   }s \le R,\\
			\sqrt{\frac{\mathcal{X}_{2}}{\mathcal{X}_{1}}}s, & \mbox{if   }s > R,
		\end{cases}R=\frac{K_{2}+\sqrt{K_{2}^{2}+4K_{1}K{3}}}{2K_{1}},
	\end{equation*}
	where $\mathcal{X}_{1}=\min\{\lambda_{\min}(P),\ell_{1}^{-1}(1+\rho_{W})\}$, $\mathcal{X}_{2}=\max\{\lambda_{\max}(P),\ell_{1}^{-1}(1+\rho_{W})\}$. Furthermore, $\|\varphi\|$ is uniformly ultimately bounded by any $\bar{d}>\sqrt{\frac{\mathcal{X}_{2}}{\mathcal{X}_{1}}}R$, with the time required for $\|\varphi\|$ to converge to the ultimate bound being
	\begin{equation*}
		T(\bar{d},s)=\begin{cases}
			0, & \mbox{if   }s \le \bar{d}\sqrt{\frac{\mathcal{X}_{2}}{\mathcal{X}_{1}}}\\
			\frac{\mathcal{X}_{2}s^{2}-(\mathcal{X}_{1}^{2}/\mathcal{X}_{2})\bar{d}^{2}}{K_{1}\bar{d}^{2}(\mathcal{X}_{1}/\mathcal{X}_{2})-K_{2}\bar{d}(\mathcal{X}_{1}/\mathcal{X}_{2})^{1/2}-K_{3}}, & \mbox{otherwise}.
		\end{cases}
	\end{equation*}
 Therefore, both the VFC-following error $\|\beta(q(t),\dot{q}(t),w(t))\|$ and the adaptive error $\|\hat{\alpha}(t)-\alpha\|$ are uniformly bounded and uniformly ultimately bounded. By Proposition \ref{prop:constraint-to-path-error} we have that the path-following error dist$(\xi(t),\mathcal{P}^{\rm hgh})$ is uniformly ultimately bounded. Consequently, by Remark \ref{markphgh}, dist$(\zeta(t),\mathcal{P})$ is uniformly ultimately bounded. \hfill$\square$. 

\end{document}